\documentclass{scrartcl}
\usepackage[margin=1in]{geometry}
\usepackage{amsmath}
\usepackage{amsfonts}
\usepackage{amssymb}
\usepackage{amsthm}
\usepackage{authblk}
\usepackage{enumitem}
\usepackage{mathtools}
\usepackage{thmtools}
\usepackage{xcolor}

\usepackage{subcaption}

\usepackage[T1]{fontenc}

\usepackage[colorlinks]{hyperref}
\usepackage[capitalize]{cleveref}
\hypersetup{linkcolor=blue,citecolor=blue,urlcolor=black}

\newtheorem{theorem}{Theorem}[section]
\newtheorem{lemma}[theorem]{Lemma}
\newtheorem{corollary}[theorem]{Corollary}
\newtheorem{conjecture}[theorem]{Conjecture}
\newtheorem{observation}[theorem]{Observation}

\newcommand{\N}{\mathbb{N}}

\newcommand{\fO}{\mathcal{O}}
\newcommand{\fF}{\mathcal{F}}

\DeclareMathOperator{\Time}{\mathbb{T}}
\DeclareMathOperator{\OPT}{OPT}
\DeclareMathOperator{\OPTZ}{\OPT_0}
\DeclareMathOperator{\OPTCT}{\OPT_{\mathrm{CT}}}
\DeclareMathOperator{\MID}{MID}
\newcommand{\alg}[1]{\textup{\texttt{#1}}}
\DeclareMathOperator{\id}{\alg{id}}

\DeclareMathOperator{\depth}{depth}
\DeclareMathOperator{\height}{height}

\DeclareMathOperator{\ub}{ub}
\DeclareMathOperator{\lb}{lb}
\newcommand{\fC}{\mathcal{C}} %

\newcommand{\logit}[1]{\log^{[#1]}}

\title{Fast decremental tree sums in forests}
\author[1,2]{Benjamin Aram Berendsohn}
\author[1,3]{Marek Sokołowski}
\affil[1]{\normalsize Max Planck Institute for Informatics, Saarbrücken, Germany}
\affil[2]{\url{benjamin.berendsohn@fu-berlin.de}}
\affil[3]{\url{msokolow@mpi-inf.mpg.de}}
\date{\vspace{-1em}}

\begin{document}
	
	\maketitle
	
	\begin{abstract}
		We study two fundamental decremental dynamic graph problems. In both problems, we need to maintain a vertex-weighted forest of size $n$ under edge deletions, weight updates, and a certain information-retrieval query. Both problems can be solved in $\fO(\log n)$ time per update/query using standard dynamic forest data structures like top trees -- even if additionally edge \emph{insertions} are allowed. We investigate whether the deletion-only problem can be solved faster.
		
		First, we consider \alg{tree-sum} queries, where we ask for the sum of vertex weights in one of the connected components (i.e., trees) in the forest. We give a data structure with $\fO(n)$ preprocessing time and $\fO(\log^* n)$ time per operation, based on a micro-macro tree decomposition (Alstrup et al., 1997).
		If the forest is \emph{unweighted} (i.e., all weights are 1 and cannot be changed), then the operation time can be improved to $\fO(1)$.
		
		Additionally, we give an asymptotically \emph{universally optimal} algorithm. More specifically, our algorithm works in the \emph{group model}, and processes $m$ operations on an~initial forest $F$ in running time $\fO( \OPT(F, m) )$. Here $\OPT(F, m)$ is the number of weight additions and subtractions that a best possible algorithm performs to handle a~worst-case instance for a~fixed initial forest $F$ and a fixed number $m$ of operations. We achieve this with a combination of the aforementioned decomposition technique, precomputation of optimal data structures for very small instances, and some insights into the behavior of $\OPT$.
		Note that even the \emph{worst-case} complexity of this algorithm remains unknown to us.
		
		Second, we consider \alg{subtree-sum} queries. Here, the forest is rooted, and a query $\alg{subtree-sum}(v)$ returns the sum of weights in the subtree rooted at~$v$. An easy reduction from the well-known \emph{prefix sum} problem shows that the general, weighted version of the problem requires $\Theta(n \log n)$ time for $n$ operations. Interestingly, we prove that the $\Omega(n \log n)$ complexity lower bound still holds even if weight updates are disallowed.
		On the other hand, we show that the unweighted version can be solved with $\fO( \tfrac{\log n}{\log\log n})$ time per operation, and this is tight.
	\end{abstract}
	
	\newpage
	
	\section{Introduction}\label{sec:intro}

	Maintaining connectivity in dynamic graphs is one of the fundamental problems in the design of algorithms, with a~plethora of major results in the area throughout the decades.
	An~efficient data structure maintaining connectivity in dynamic \emph{forests} was given by Sleator and Tarjan~\cite{SleatorTarjan1983}: Their \emph{link/cut trees} support connectivity queries under edge insertions and removals in an~$n$-vertex forest in $\fO(\log n)$ time.
	They also show that their data structure supports \emph{weights} of vertices and various queries related to these weights in the same time complexity, such as the sum of weights of vertices in a~tree component containing a~given vertex, or the sum of weights of vertices in a~subtree rooted at a~specified vertex.
	The optimality of this data structure -- even in the unweighted setting -- was only shown much later by P{\u{a}}tra{\c{s}}cu and Demaine~\cite{PatrascuDemaine2006}.
	In the more general setting of arbitrary dynamic graphs, the asymptotic behavior of an~optimal data structure is still not fully understood, with the currently best algorithm of Huang et al.~\cite{HuangHuangKPT2023} being a~factor of $\Theta((\log \log n)^2)$ away from the logarithmic lower bound.
	
	Aiming to break through the logarithmic operation lower bound, we may decide to restrict the space of possible updates to the graph.
	The \emph{incremental} setting, in which edges may only be inserted to the graph, is precisely the \emph{disjoint set union} problem, which was already studied by Galler and Fischer~\cite{GallerFischer1964} and subsequently improved: Hopcroft and Ullman gave a~data structure supporting insertions and connectivity queries in $\fO(\log^* n)$ time per operation~\cite{HopcroftUllman1973}, with the complexity analysis later refined by Tarjan to $\fO(\alpha(n))$ per operation~\cite{Tarjan1975}.
	Here, $\log^* n$ and $\alpha(n)$ denote two very slow-growing functions, namely the \emph{iterated logarithm} and the \emph{inverse Ackermann function}.
	Again, the presented data structures can be easily adapted to support weights of vertices and report aggregate information on the connected components of the graphs.
	The bound of $\fO(\alpha(n))$ was later shown to be optimal in various computational models by Tarjan~\cite{Tarjan1979} and Fredman and Saks~\cite{FredmanSaks1989}, even under the promise that the graph always remains a~forest.

	In this work, we focus on the \emph{decremental} setting, where a~data structure is given an~initial graph and edges may only be removed from the graph.
	This restriction has also received considerable attention over the years, and the decremental connectivity problem has been solved optimally: A~data structure processing $m$ operations in an~$n$-vertex forest in total $\fO(n + m)$ time was given by Alstrup, Secher and Spork~\cite{AlstrupSecherEtAl1997}, with the result later strengthened to planar graphs by Łącki and Sankowski~\cite{LackiSankowski2017}.
	However, interestingly, even the decremental forest data structure employs various bit-optimization tricks to achieve the optimal time complexity, which renders it unable to handle weights of vertices efficiently.
	This suggests a~natural question: How efficiently could we implement a~data structure for \emph{weighted} decremental forests that can also report summaries of tree components -- such as the sum of the weights of the vertices in the component -- dynamically?
	And can such a~data structure also handle more complex queries, such as the sum of the weights of the vertices in a~rooted subtree of a~specified tree component?

	\paragraph{Setting.} We now formally introduce the problem. The task is to maintain a vertex-weighted \emph{rooted} forest $(F,w)$, where weights come from some commutative group $G$, under some or all of the following operations:
	\begin{itemize}[nosep]
		\item $\alg{cut}(v)$: Delete the edge between $v$ and its parent (we assert that this edge exists).
		\item $\alg{update-weight}(v, x)$: Set $w(v) \gets x$.
		\item $\alg{tree-sum}(v) \rightarrow x$: Return the sum $x$ of weights in the tree containing $v$.
		\item $\alg{subtree-sum}(v) \rightarrow x$: Return the sum $x$ of weights in the subtree $F_v$ rooted at $v$.
	\end{itemize}

	We are mostly interested in the \emph{group model} of computation, where the weights of the vertices are considered to be opaque objects.
	The implementation of a~data structure may manipulate the weights (elements of the group) by adding them together or subtracting one from another, but no other operations -- such as comparisons or the extraction of the representation of the weight in memory -- are allowed.

	All operations listed above can be implemented in $\fO(\log n)$ worst-case time in the group model with well-known efficient fully dynamic forest data structures such as aforementioned link/cut trees~\cite{SleatorTarjan1983} or top trees~\cite{AlstrupHolmLT2005}, so any sequence of $m$ operations can be processed in $\fO((n + m) \log n)$ time.
	The question remains: Can (some or all) of the operations above be implemented more efficiently, knowing that the edges may only be removed from the forest?

	\subsection{Results and techniques}
	
	In the following, we formally present our results and give basic ideas on how to prove them. In the informal discussion following each result, we will assume for convenience that the number of operations $m$ is roughly equal to the number $n$ of vertices in the input forest. Thus the running times in the discussion will always be the total of the initialization time and the time required to process all $m \approx n$ operations. We will also assume the forest given initially is a tree.
	
	\paragraph{Tree-sums with recursive cluster decompositions.}
	Our first result concerns the basic \alg{tree-sum} problem.
	
	\begin{restatable}{theorem}{restateTreeSum}\label{p:tree-sum}
		There is a data structure that maintains a weighted rooted forest $(F,w)$ on $n$ vertices under \alg{cut}, \alg{update-weight}, and \alg{tree-sum}, with running time $\fO( (m+n) \log^* n)$ for $m$ operations.
	\end{restatable}
	
	We prove \cref{p:tree-sum} in \cref{sec:tree-sum} using the well-known technique of \emph{micro-macro decomposition} (defined in \cite{AlstrupSecherEtAl1997}, though similar ideas already exist in~\cite{GabowTarjan1985}), although we adapt it to our purposes and call it \emph{cluster decomposition}. Here, we define this decomposition for \emph{rooted trees}, which simplifies the presentation considerably.
	The basic idea is to partition the input tree into $\Theta(\tfrac{n}{\log n})$ \emph{clusters}, each of size $\fO(\log n)$, and contract each of these parts into either a single vertex or two vertices connected by an~edge to obtain a \emph{cluster tree} of size $\Theta(\tfrac{n}{\log n})$.
	
	Recall a~basic $\fO(n \log n)$ algorithm; call it $\mathsf{Alg}_1$.
	An $\fO(n \log\log n)$ algorithm $\mathsf{Alg}_2$ can be obtained as follows: Construct the cluster tree in linear time, run $\mathsf{Alg}_1$ on this cluster tree, and the same algorithm on each cluster. The running time for the cluster tree algorithm is $\fO(\tfrac{n}{\log n} \log \tfrac{n}{\log n}) = \fO(n)$, and the running time for each cluster with $k \le \log n$ vertices is $\fO(k \log k)$, for a total of $\fO(n \log\log n)$. This approach can be iterated: an~algorithm $\mathsf{Alg}_t$ is implemented the same way, only that $\mathsf{Alg}_{t-1}$ is used for each of the clusters. Iterating this approach $\log^* n$ times gives an algorithm with $\log^* n$ ``levels'', each of which takes $\fO(n)$ time in total; thus, the total running time of $\mathsf{Alg}_{\log^* n}$ is $\fO(n \log^* n)$.
	
	Of course, the above presentation is heavily simplified. We will note one interesting detail. Performing a \alg{cut} within a cluster may affect the weights stored in the cluster tree, necessitating an \alg{update-weight} operation on the cluster tree. Hence, the number of operations in the cluster tree can be $\Theta(n)$, instead of $\Theta(\tfrac{n}{\log n})$, as we assumed above. The solution to this is to use an algorithm where each \alg{update-weight} operation takes only $\fO(1)$ time (whereas \alg{cut}s each take $\Theta(\log n)$ amortized time). Note that standard dynamic tree solutions require $\Theta(\log n)$ per operation, and thus are not suitable for this purpose. Instead, we design a simple decremental-only algorithm $\mathsf{Alg}_1$ with this property.
	
	\paragraph{Unweighted and 0-1-weighted tree-sums.}
	For unweighted forests, we can improve the above algorithm (\cref{sec:tree-size}). 
	In fact, our solution works when the weight function only takes values zero and one.\footnote{The underlying group is the set of integers, so if all weights are one, we recover the unweighted tree-sum problem.}
	
	\begin{restatable}{theorem}{restateTreeSize}\label{p:tree-0-1-sum}
		There is a data structure in the Word RAM model of computation that maintains a 0-1-weighted rooted forest $(F,w)$ with $n$ vertices under \alg{cut}, \alg{update-weight} and \alg{tree-sum} operations, in $\fO(m+n)$ time for $m$ operations.
	\end{restatable}
	
	We describe the algorithm for the unweighted case.
	The idea is to apply the cluster decomposition above only twice; then the bottom-level clusters form a decomposition of the input tree into $\fO(\tfrac{n}{\log\log n})$ small trees of size at most $k \le \fO(\log\log n)$. Note that the number of distinct (labeled) trees of size at most $k$ is no more than $k^{\fO(k)} \ll \tfrac{n}{\log\log n}$.
	Thus, we can essentially precompute all answers for every possible such tree, and then do a lookup for each cluster.
	
	This technique is commonly used in algorithm design~\cite{ArlazarovDinicEtAl1970,BerkmanVishkin1994,BenderFarachColton2000,BenderFarachColton2004}. For a \emph{data structure} this is not quite as straightforward, since the queries are not known at the beginning. Still, we can compute something akin to a \emph{state graph} that stores all necessary information, and can be navigated efficiently.
	
	\paragraph{An optimal weighted tree-sum algorithm.}
	We can take the idea of precomputation even further (\cref{sec:tree-sum-optimal}), based on a~technique of \emph{decision trees} used by Pettie and Ramachandran~\cite{PettieRamachandran2002} (also see e.g.~\cite{Larmore1990,Matousek1992} for different applications of this approach). Perform the cluster decomposition three times, obtaining trees of size at most $k \le \fO(\log\log\log n)$. For these very small trees, we can precompute the \emph{optimal data structure} in time $o(n)$ (\cref{sec:opt-small}). If $\OPT(n,m)$ is the optimal total running time for $m$ operations on a forest with $n$ vertices, then our algorithm has running time
	\[ \fO(n+m) + \sum_i \OPT(n_i,m_i), \]
	where $n_i \le k$ is the number of vertices in the $i$th cluster, and $m_i$ is the number of operations applied to that cluster.
	Note that $\sum n_i \leq n$ and $\sum m_i \leq m$.
	
	With some effort (\cref{sec:opt-superadd}), we can argue that the above formula is upper bounded by $\fO( n + m + \OPT(n,m))$, since the function $\OPT$ satisfies a weak, asymptotic form of \emph{super-additivity}: $\sum_{i} \OPT(n_i, m_i) \leq \fO(\OPT(\sum_i n_i,\, \sum_i m_i) + \sum_i n_i)$ for any pair of positive integer sequences $\{n_i\}$, $\{m_i\}$.
	Interestingly, even this weak variant of super-additivity is quite challenging to prove.
	The proof proceeds via a~series of reductions involving data structures for \alg{tree-sum} that are initialized with forests with \emph{zero weights}.
	The rationale is that $\OPTZ(n, m)$ -- the function denoting the optimal total running time for $m$ operations on a~\emph{zero-initialized} forest with $n$ vertices -- can be directly shown to be super-additive, i.e., $\sum_{i} \OPTZ(n_i, m_i) \leq \OPTZ(\sum_i n_i,\, \sum_i m_i)$.

	We finally obtain

	\begin{restatable}{theorem}{restateTreeSumOptReal}\label{p:tree-sum-opt}
		There exists a data structure with running time $\fO( \OPT(n,m) + n + m )$ for any forest on $n$ vertices and any sequence of $m$ operations.
	\end{restatable}
	
	Note that our data structure is as fast (up to constant) as any other data structure, even if the latter has $n$ and $m$ hard-coded into it. It turns out we can prove something even stronger: Our data structure is optimal even compared to other data structures with a hard-coded \emph{initial forest} $F$. For such data structures, the only non-hard-coded parts of the input are the initial weights and the sequence of operations.
	This is precisely the notion of \emph{universal optimality}, popularized by works of Haeupler, Wajc, and Zuzic~\cite{HaeuplerWajcZuzic2021} and Haeupler, Hladík, Rozhon, Tarjan, and Tetek~\cite{HaeuplerHladikRozhonTarjanTetek2024}, and we adopt it here as well.
	Universally optimal algorithms have been recently studied for various \emph{static} problems, such as Single-Source Shortest Paths~\cite{HaeuplerWajcZuzic2021,HaeuplerHladikRozhonTarjanTetek2024,HoogRotenbergRutschmann2025c}, planar convex hull~\cite{AfshaniBarbayChan17,HoogRotenbergRutschmann2025a} and sorting~\cite{HaeuplerHIRTT2025,HoogRotenbergRutschmann2025b}.
	A recent work of the first author presents a universally optimal data structure for the related \emph{tree-min} problem~\cite{Berendsohn2026}.
	
	We now discuss some details of \cref{p:tree-sum-opt}, in particular in comparison with Pettie and Ramachandran's famous optimal minimum spanning tree (MST) algorithm~\cite{PettieRamachandran2002}. They similarly reduce a large instance to many smaller instances, though their reduction is much more complex.
	On the other hand, for them, computing the optimal algorithm for small instances is quite straightforward, since every MST algorithm in the comparison model can be represented as a decision tree. Thus, they only need to enumerate all decision trees for the given small graph, and choose the best correct one.
	
	In contrast, the tree-sum problem is an \emph{online} problem in the \emph{group model}, which presents some interesting challenges. For example, testing correctness of a decision tree computing the MST of a~static $k$-vertex graph is easy (just check all possible graphs, and for each graph, all possible input weight orders), but testing correctness of an algorithm (or a~data structure) in the group model is harder. Another interesting problem is that the number operations applied to a single cluster is initially unknown, and can be much larger than the size of the cluster. Dealing with these problems requires some new insights.\footnote{In a different paper, Pettie and Ramachandran also give a provably optimal algorithm for an online data structure problem~\cite[Appendix B]{PettieRamachandran2005}. This is in the comparison model, so some of our challenges do not apply. Additionally, there appears to be an error in that paper that hides some of the complexity; see the discussion around \begin{NoHyper}\cref{p:opt-small}\end{NoHyper} in \cref{sec:opt-small-compute}.}

	Finally, Pettie and Ramachandran note that an explicit provably optimal algorithm for MST was actually known before. An argument of Jones~\cite{Jones1997} transforms any optimal verification algorithm for a~computational problem (which is known for MST) into an optimal algorithm.
	Interestingly, Jones's argument does not work for the online data structure problems like ours. Thus, our optimality result in itself is novel.
	In any case, this approach yields massive constants in the running time and does not imply anything about the \emph{decision tree complexity} of MST. In contrast, the running time of Pettie and Ramachandran's algorithm matches the decision tree complexity. \Cref{p:tree-sum-opt} similarly matches the optimal number of additions and subtractions, up to constant factors.
	That means that any algorithm that only performs a linear number of additions and subtractions in total (with arbitrary running time) implies a data structure with linear total running time.

	\paragraph{Subtree-sums.}
	Finally, we consider the harder subtree-sum problem.
	In its full generality, with support for \alg{subtree-sum}, \alg{update-weight}, and \alg{cut}, it is a generalization of the well-known partial-sum problem~\cite{BentleyMaurer1980,Dietz1989,Fenwick1994,Fredman1982,FredmanSaks1989,Yao1985}. In the partial-sum problem, we are tasked to maintain an array of weights under the operations of computing the sum of a prefix and changing an entry. Clearly, we can interpret the array as a degenerate path-like rooted tree. Then partial sums are \alg{subtree-sum} operations and entry changes are \alg{update-weight} operations (\alg{cut} operations are not used). For the partial-sum problem, an $\Omega(n\log n)$ lower bound was shown by Pătraşcu and Demaine~\cite{PatrascuDemaine2006}, so we have:

	\begin{observation}\label{p:sumtree-sum-easy-lb}
		Each data structure that maintains a weighted forest with $n$ vertices under the operations \alg{subtree-sum} and \alg{update-weight} requires $\Omega( n \log n)$ time for $n$ operations, even in the Word RAM model of computation and assuming $G = \mathbb{Z}$.

	\end{observation}
	
	Much more interesting is the case where we do not allow \alg{update-weight} operations, only \alg{subtree-sum} and \alg{cut}.
	In \cref{sec:sumtree-sum-weights-with-cut-lb}, we show a~randomized reduction from the lower bound of~\cite{PatrascuDemaine2006}, asserting the $\Omega(n \log n)$ lower bound in this restricted setting:
	\begin{restatable}{theorem}{restateSubtreeSumHard}\label{p:sumtree-sum-weights-with-cut-lb}
		Each data structure that maintains a weighted forest with $n$ vertices under the operations \alg{subtree-sum} and \alg{cut} requires $\Omega( n \log n)$ time for $\Theta(n)$ operations in the Word RAM model of computation and assuming $G = \mathbb{Z}$.
	\end{restatable}
	On a high level, we design a~hard instance for the problem as follows.
	We show in \cref{sec:adapted-prefix-sum-lb} that the $\Omega(n \log n)$ lower bound of Pătraşcu and Demaine also holds under the assumption that the underlying array contains $\sqrt{n}$ integers ranging from $0$ to $\Theta(\sqrt{n})$, with each element of the array updated $\Theta(\sqrt{n})$ times.
	The weighted forest in our proof consists then of a~spine of $\sqrt{n}$ vertices -- each representing an~element of the array -- with $\sqrt{n}$ leaves of various weights attached to each vertex of the spine, so that a~partial sum of $i$ initial elements of the array can be inferred from the \alg{subtree-sum} of the subtree rooted at the $i$th vertex of the spine.
	We also show a Las Vegas randomized subroutine that efficiently translates an~update of the $i$th element of the array to a~sequence of $\fO(1)$ \alg{cut}s of leaves that are still attached to the $i$th vertex of the spine.
	This way, any sequence of $\Theta(n)$ operations in a \alg{partial-sum} data structure is interpreted as a~sequence of $\Theta(n)$ operations in the \alg{subtree-sum} structure, and the $\Omega(n \log n)$ lower bound for the former implies the same lower bound for the latter.

	With \cref{p:sumtree-sum-easy-lb,p:sumtree-sum-weights-with-cut-lb} in mind, we study the \emph{unweighted} variant of the problem (or, more precisely, the \emph{0-1-weighted} variant with no weight updates). Here, we are able to obtain an improved (amortized) running time of $\fO(\tfrac{\log n}{\log\log n})$ per operation (\cref{sec:subtree-size-ub}). Again, we state the slightly more general case of 0-1 weights: %
	
	\begin{restatable}{theorem}{restateSubtreeSize}\label{p:subtree-size}
		There is a data structure in the Word RAM model maintaining a 0-1-weighted forest with $n$ vertices under the operations \alg{subtree-sum} and \alg{cut}, with $\fO((m+n) \tfrac{\log n}{\log\log n})$ total time for $m$ operations.
	\end{restatable}
	
	\Cref{p:subtree-size} again uses recursive cluster decompositions, but with larger clusters: Instead of $\fO(\tfrac{n}{\log n})$ clusters of size $\fO(\log n)$, we use roughly $k = \sqrt{\log n}$ clusters of size $\fO(\tfrac{n}{k})$. Recursively applying this decomposition until clusters have constant size yields a decomposition of depth $\fO(\log_k n) = \fO(\tfrac{\log n}{\log\log n})$.
	Each cluster tree obtained in this process has size at most $k$, which notably implies that it can be encoded with $\fO(\log n)$ bits, which fits into a single word. This allows us to combine a delayed-updating technique of Dietz~\cite{Dietz1989} with some precomputation to obtain a fast data structure for these very small cluster trees.

	Note that the paper of Dietz~\cite{Dietz1989} cited above solves a~variant of our problem where the initial tree is a 0-1-weighted path updated by \alg{update-weight} rather than \alg{cut}, with the same $\fO((m + n)\tfrac{\log n}{\log\log n})$ total time; this data structure is known to be optimal due to Fredman and Saks~\cite{FredmanSaks1989}.
	On the other hand, the path variant of \cref{p:subtree-size} (reporting \alg{subtree-sum}s in an~unweighted or 0-1-weighted path with $n$ vertices updated by \alg{cut}) is much easier -- it can be solved in $\fO(m+n)$ time, e.g.\ via a~reduction to the \emph{marked ancestor problem}~\cite{AlstrupHusfeldtEtAl1998a}.
	In contrast, we show in \cref{sec:subtree-size-lb} that for general forests, our problem is as hard as the path variant with weight updates.
	
	\begin{restatable}{theorem}{restateSubtreeSizeLB}\label{p:subtree-size-lb}
		Each data structure in the Word RAM model that maintains an unweighted forest with $n$ vertices under the operations \alg{subtree-sum} and \alg{cut} requires $\Omega( n \tfrac{\log n}{\log\log n})$ time for $n$ operations.
	\end{restatable}

	In other words, the data structure of \cref{p:subtree-size} is asymptotically optimal.

	\subsection{Outlook to the offline setting}
	We can consider the decremental \alg{tree-sum} problem in an~\emph{offline} setting, where the entire sequence of $m$ operations is provided to the algorithm together with the initial weighted $n$-vertex forest at the time of initialization.
	It turns out that the offline variant of the problem can be solved in $\fO(n + m)$ time: Simulating the sequence of operations backwards, we reduce \alg{tree-sum} to the incremental variant -- that is, the disjoint set union problem -- where the \emph{merge forest} (the forest containing an~edge $uv$ for every merge of sets containing vertices $u$ and $v$) is known in advance.
	This simplified variant of the problem is solvable in linear time in the Word RAM model by a~data structure of Gabow and Tarjan~\cite{GabowTarjan1985}.
	Moreover, reporting \alg{tree-sum}s requires additionally only $n - 1$ operations in the group $G$: Whenever two sets (i.e., vertex sets of two tree components of a~forest) are merged, the tree sum of the resulting tree components is simply the sum of the tree sums of the components being merged.
	
	This discussion unveils an~interesting phenomenon: In the online incremental setting of \alg{tree-sum}, it is the maintenance of connected tree components that is strictly more computationally expensive (requiring on average $\Omega(\alpha(n))$ time per edge insertion) than the additional bookkeeping of the tree sums ($\fO(1)$ time per insertion).
	However, the opposite appears to be the case in the decremental setting: The tree components can be maintained in amortized constant time per edge removal, but it is not clear at all whether tree sums can be tracked this efficiently.
	
	\section{Preliminaries}
	\label{sec:prelims}

	For $k \in \N$, define $[k] = \{1, 2, \ldots, k\}$.
	All logarithms in this work are binary, unless specified otherwise.
	For convenience, assume that $\log x = 0$ for all $x \leq 1$.
	Then for $k \in \N_+$, we define the $k$-fold logarithmic function $\log^{[k]} x$: $\log^{[1]} x = \log x$ and $\log^{[k]} x = \log(\log^{[k-1]} x)$ for $k > 1$.
	Finally, the iterated logarithm, $\log^* x$, is defined as the smallest $k$ such that $\log^{[k]} x \leq 1$.
	
	All trees and forests in this work are rooted.
	In a~tree $T$ and vertices $u$ and $v$, we say that $u$ is an~ancestor of $v$ if $u$ lies on the unique path from the root of $T$ to $v$.
	In this case, we also say that $v$ is a~descendant of $u$.
	Strict ancestors and descendants are defined analogously, only that we additionally require that $u \neq v$.
	For a tree $T$ and its vertex $v$, we denote by $T_v$ the subtree of $T$ rooted at $v$, i.e., the subtree induced by all descendants of $v$. The \emph{depth} of a vertex $v$, denoted by $\depth_T(v)$, is its distance to the root (0 for the root itself), and the height of $T$ is the maximum depth of a vertex.
	
	We assume the standard Word RAM model of computation with memory cells (\emph{words} or \emph{registers}) of size $b = \Theta(\log n)$~\cite{FredmanWillard1990}.
	Formally, the memory is represented by an~array $R$ of $b$-bit integers, with access to standard arithmetic and bitwise operations, as well as comparisons and indirect accesses (of the form $R[R[i]] \gets R[j]$).
	Each such access takes constant time.

	Most of the algorithms in this paper work in the \emph{group model}.
	In this model, some memory cells may contain -- instead of a~$b$-bit integer -- an~element of a~commutative group $G$.
	The elements of $G$ can be manipulated in constant time via addition and subtraction, and the comparisons between the elements of $G$ (including the zero-comparisons of the elements of $G$) are disallowed.
	In particular, it is forbidden to examine or modify the memory representation of any memory cell containing an~element of $G$.
	The weights of vertices of $F$ come from the group $G$.
	For convenience, we will denote by $W[v]$ the memory cell occupied by the current weight $w(v)$ of a~vertex $v$.

	If $F$ is a~forest with a~weight function $w : V(F) \to G$, then for $U \subseteq V(F)$ we define $w(U) = \sum_{v \in U} w(v)$.
	Similarly, if $T$ is a~(sub)tree of $F$, then $w(T) = w(V(T))$.
	If $S \subseteq V(F)$, then we denote by $F[S]$ the subgraph of $F$ induced by $S$, and by $w|_S : S \to G$ the restriction of the weight function $w$ to the vertices in $S$.
	
	\paragraph{Decremental connectivity.}
	Our data structures rely heavily on the Word RAM data structures supporting connectivity queries in decremental forests.
	Such a~data structure, when initialized with a~forest $F$, supports the following updates and queries:

	\begin{itemize}[nosep]
		\item $\alg{cut}(v)$: Delete the edge between $v$ and its parent (we assert that this edge exists).
		\item $\alg{root}(v) \rightarrow r$: Return the root of the tree containing $v$.
		\item $\alg{connected}(u, v) \rightarrow \textsf{bool}$: Indicate whether $u$ and $v$ are in the same component of $F$. %
		\item $\alg{ancestor}(u, v) \rightarrow \textsf{bool}$: Indicate whether $u$ is an ancestor of $v$ in $F$ (which implies they are connected).
	\end{itemize}

	As mentioned before, in the Word RAM model, there exists a~linear-time data structure for decremental connectivity of forests supporting \alg{cut} and \alg{root}~\cite{AlstrupSecherEtAl1997}.
	Therefore:
	
	\begin{lemma}[{\cite{AlstrupSecherEtAl1997}}]\label{p:roots}
		There is a Word RAM data structure that maintains a forest with $n$ vertices, and supports \alg{cut}, \alg{root}, \alg{connected}, and \alg{ancestor} in $\fO(n+m)$ time for $m$ operations.
	\end{lemma}
	\begin{proof}
		The operation $\alg{connected}(u,v)$ is easily implemented by comparing $\alg{root}(u)$ and $\alg{root}(v)$. For \alg{ancestor}, observe that if $F^0$ is the original forest and $F$ is the current forest, then a vertex $u$ is an ancestor of a vertex $v$ if and only if $u$ is an ancestor of $v$ in $F^0$ and the two vertices are connected in $F$. The ancestor queries in $F^0$ can be answered in constant time after linear preprocessing by consulting the pre-order and post-order traversals of $F^0$.%
	\end{proof}
	
	\paragraph{Binarization and auxiliary vertices.}
	In our algorithms, it will be convenient to transform the input trees into \emph{binary} trees, in which each vertex has at most two children.
	It is well-known that this can be done by transforming high-degree vertices into binary trees, which introduces a linear number of additional vertices.
	For most applications, the fact that the number of vertices increases by only a constant factor is sufficient to retain desired running times.
	However, for our optimal algorithm in \cref{sec:tree-sum-optimal}, inserting additional vertices causes technical problems. To address this, we introduce the concept of \emph{auxiliary vertices}.
	
	Let $F$ be the input forest for one of our data structures. If a vertex $v$ in $F$ is marked as auxiliary, then no operations may be applied to $v$ (e.g., $\alg{cut}(v)$ and $\alg{update-weight}(v, \cdot)$ are illegal).
	Moreover, if a weight function is associated to $F$, then $w(v) = 0$.
	In the interest of brevity, we still write just $F$ for a forest with auxiliary vertices; it will be made clear at the start of each relevant section whether algorithms allow auxiliary vertices or not.
	
	Using auxiliary vertices, we can transform any forest into a binary forest that is equivalent for our data structure problems.

	\begin{theorem}\label{p:binarize}
		Let $(F,w)$ be a weighted forest. Then, in linear time, we can compute a \emph{binary} weighted forest $(F',w')$, such that $V(F) \subseteq V(F')$, $|V(F')| \le 2|V(F)|$, every sequence of operations \alg{cut}, \alg{update-weight}, \alg{tree-sum}, and \alg{subtree-sum} is legal in $(F, w)$ if and only if it is legal in $(F', w')$, and every such sequence yields the same results on $(F,w)$ and $(F',w')$.
		
		Moreover, $w'(v) = w(v)$ for each $v\in V(F)$ and $w'(v) = 0$ for each $v \in V(F') \setminus V(F)$.
	\end{theorem}
	\begin{proof}
		For every vertex $v \in V(F)$ with $d_v \geq 3$ children, we replace $v$ with a rooted path $P_v$ consisting of a~root $v$ and $d_v - 1$ auxiliary vertices of weight $0$. Each child of $v$ then becomes a child of a different vertex in $P_v$.
		This transformation can be done in linear time, and it clearly satisfies all requirements of the theorem.
	\end{proof}
	
	\subsection{The cluster decomposition}

	We now formally describe a~decomposition of a~binary tree $T$ into small disjoint well-structured subgraphs of $T$, which we call the \emph{cluster decomposition} of $T$.
	The decomposition is inspired by the work of Alstrup et al.~\cite{AlstrupSecherEtAl1997}, however our variant of the decomposition will be particularly convenient to use in the upcoming data structures.

	Let $T$ be a tree, and let $C \subseteq V(T)$ induce a connected subtree of $T$. A vertex $v \in C$ is called an \emph{upper boundary vertex} if the parent of $v$ is not contained in $C$, or $v$ is the root of $T$. Note that each connected set $C \subseteq V(T)$ has precisely one upper boundary vertex, which we denote by $\ub(C)$. A vertex in $C$ is called a \emph{lower boundary vertex} if it has at least one child that is not contained in $C$. We call $C$ a \emph{cluster} if it has at most one lower boundary vertex, and moreover, no child of the lower boundary vertex is in $C$. Let $\lb(C)$ denote the lower boundary vertex of a cluster $C$, if it exists, or let $\lb(C) = \bot$ otherwise. Note that a single vertex is always a valid cluster.
	
	\begin{figure}
		\centering
		\includegraphics{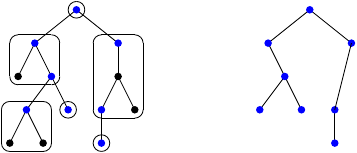}
		\caption{A cluster decomposition (left) and the associated cluster tree (right)}\label{fig:cluster-decomp}
	\end{figure}
	
	A \emph{cluster decomposition} (see \cref{fig:cluster-decomp}) of a tree $T$ is a partition of $V(T)$ into a set $\fC$ of disjoint clusters. The \emph{cluster tree} of $\fC$ is a tree $S$ where $V(S)$ are all boundary vertices of the clusters in $\fC$. A vertex $u \in V(S)$ is a child of $v \in V(S)$ if one of the following is true:
	\begin{itemize}[nosep]
		\item There is a cluster $C \in \fC$ with $u = \lb(C)$ and $v = \ub(C)$, or
		\item $u$ and $v$ are in different clusters and $u$ is a child of $v$ in $T$ (then $u$ must be an upper boundary vertex and $v$ must be a lower boundary vertex).
	\end{itemize}
	
	Note that every vertex in $V(S)$ has a parent, except the root $r$ of $T$. Thus, the cluster tree is indeed a (rooted) tree.
	Note that the cluster tree contains no auxiliary vertices, regardless of whether $T$ contains auxiliary vertices or not.
	
	\begin{lemma}\label{p:cluster-decomp}
		Let $T$ be a binary tree with $n$ vertices, and let $k \le n$. Then, in linear time, we can compute a cluster decomposition of $T$ into at most $6\tfrac{n}{k}$ clusters, each of size at most $k$, along with the corresponding cluster tree.
	\end{lemma}
	\begin{figure}
		\begin{subfigure}{.4\linewidth}
			\centering
			\includegraphics{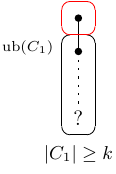}
			\includegraphics{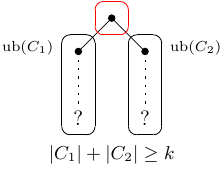}
			\caption{Two variants of case \ref{case:clusters-large}}
		\end{subfigure}%
		\begin{subfigure}{.3\linewidth}
			\centering
			\includegraphics{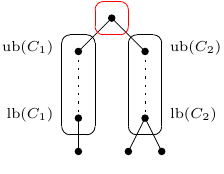}
			\caption{Case \ref{case:clusters-lb}}
		\end{subfigure}%
		\begin{subfigure}{.3\linewidth}
			\centering
			\includegraphics{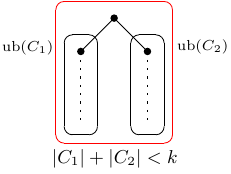}
			\caption{Case \ref{case:clusters-merge}}
		\end{subfigure}%
		\caption{Illustration of the proof of \cref{p:cluster-decomp}. The returned cluster is colored red.}\label{fig:cluster-decomp-algo}
	\end{figure}
	\begin{proof}
		The algorithm works by growing and merging clusters in a bottom-up fashion. We formulate it as a recursive algorithm $\alg{construct}$, which, applied to a vertex $v$, partitions the rooted subtree $T_v$ into clusters, and returns the cluster $C$ containing $v$. Each cluster stores $\ub(C)$, $\lb(C)$ and $|C|$, so reconstructing the cluster tree is easy afterwards.
		
		Consider a call $\alg{construct}(v)$. We first apply \alg{construct} on the up to two children of~$v$. Let $\fC'$ be the set of up to two returned clusters. Consider these three cases (see \cref{fig:cluster-decomp-algo}):
		\begin{enumerate}[nosep, label = (\alph*)]
			\item If the total size of the clusters in $\fC'$ is at least $k$, then we create a new singleton cluster $C$ with $V(C) = \{v\}$, $\ub(C) = \lb(C) = v$, which we return.\label{case:clusters-large}
			\item If there are two clusters $C_1, C_2 \in \fC'$, $\lb(C_1) \neq \bot$ and $\lb(C_2) \neq \bot$, then we return a~singleton cluster as above.\label{case:clusters-lb}
			\item Otherwise, we merge the clusters in $\fC'$ with $v$. We create a new cluster $C$ with $V(C)$ containing $v$ and all vertices in $\fC'$ and set $\ub(C) = v$. We further set $\lb(C) = u$ if $\lb(C') = u$ for some $C' \in \fC'$ (note that $u$ is unique by assumption), or $\lb(C)=\bot$ otherwise. (Note that if $\fC'$ is empty, then this creates a singleton cluster as in the previous cases, except that $\lb(C) = \bot$.)\label{case:clusters-merge}
		\end{enumerate}
		
		It is easy to see that this creates a valid cluster decomposition, and that all clusters have size at most $k$. We now show that the total number of clusters is at most $6\tfrac n k$.
		
		Consider all times a cluster is finished by the algorithm. This can happen in three different ways. First, in the root call, a single root cluster is returned at the end. Second, in \ref{case:clusters-large}, we finish one or two clusters of total size at least $k$. This can happen at most $\tfrac{n}{k}$ times, limiting the number of such clusters to $2\tfrac{n}{k}$.
		
		Third and finally, in \ref{case:clusters-lb}, we finish exactly two clusters $C_1, C_2$ of arbitrary size.
		Both these clusters have a lower boundary vertex, which means that are not leaves in the cluster tree $S$. Thus, we can already deduce that the number of \emph{leaves} in $S$ is at most $2\tfrac n k$ from the above.
		Additionally, $C_1$ and $C_2$ are children of the same cluster (which was just newly created) in $S$. Thus, each instance of this third case happening corresponds to an inner vertex of $S$ with exactly two children. The number of such vertices is bounded by the number of leaves in $S$, minus one. Thus, the total number of clusters is at most $1 + 2\tfrac n k + 2 \cdot (2\tfrac n k-1) \le 6\tfrac n k$.
	\end{proof}
	
	Both of our main algorithms use a cluster decomposition $\fC$ of the input tree. The decomposition $\fC$ itself is never modified, even when clusters become disconnected. However, we will modify the cluster tree by removing edges; for this, the following definition will be useful.
	
	Let $T$ be a tree, let $\fC$ be a cluster decomposition of $T$, and let $S$ be the corresponding cluster tree. For every forest $F$ obtained by removing some set of edges from $T$, we define the \emph{cluster forest induced by $F$} as the forest $G$ with $V(G) = V(S)$, such that $u$ is a child of $v$ in $G$ if:
	\begin{itemize}[nosep]
		\item There is a cluster $C \in \fC$ with $u = \lb(C)$ and $v = \ub(C)$, and $u, v$ are connected in~$F$, or
		\item $u$ and $v$ are in distinct clusters and $u$ is a child of $v$ in $F$.
	\end{itemize}
	
	Note that $G$ can be obtained from $S$ by removing edges. More specifically, consider a sequence of edge removals within $T$ that produces $F$. To construct $G$, we do the following: Whenever an edge between the boundary vertices of two clusters is removed, the corresponding two vertices in $G$ are likewise disconnected. Whenever an edge within a cluster $C$ is removed, if this disconnects $\lb(C)$ and $\ub(C)$, the corresponding edge in $G$ is deleted.
	Since connectivity can be maintained in constant time per operation (\cref{p:roots}), it is easy to construct the following data structure:
	
	\begin{lemma}\label{p:maintain-cluster-forest}
		Suppose we are given a tree $T$ on $n$ vertices, and a cluster decomposition $\fC$ of $T$. There is a data structure that maintains the induced cluster forest under (up to $n-1$) edge deletions in $\fO(n)$ total time.
	\end{lemma}
	
	\section{Tree sums in \texorpdfstring{$\fO(n \log^* n)$}{O(n log* n)} time}\label{sec:tree-sum}
	
	In this section, we prove:
	
	\restateTreeSum*
	
	We start with a simple $\fO(n \log n + m)$-time data structure (\cref{sec:tree-sum-logn}), which we then iterate using cluster decompositions (\cref{sec:tree-sum-it}).
	
	For technical reasons, our data structures need to support a slightly modified \alg{cut} operation: $\alg{cut-report}(v)$ performs a $\alg{cut}(v)$, but additionally returns the results of $\alg{tree-sum}(v)$ and $\alg{tree-sum}(u)$, where $u$ is the (former) parent of $v$. This will be useful to avoid extra calls to \alg{tree-sum} that complicate the running time analysis. 
	
	In this section, let a \emph{tree-sum data structure} refer to a data structure that maintains a \emph{binary} weighted forest, potentially with auxiliary vertices, under \alg{tree-sum}, \alg{update-weight}, and \alg{cut-report}.
	By \cref{p:binarize}, this also implies an analogous data structure for general weighted forests.

	\subsection{A simple algorithm with constant-time queries and weight updates}\label{sec:tree-sum-logn}

	As mentioned in the introduction, well-known dynamic forest data structures can solve the tree-sum problem with $\fO(\log n)$ time per operation. The data structure presented next takes $\Theta(\log n)$ (amortized) time per \alg{cut}, but only $\fO(1)$ time for \alg{tree-sum} and \alg{update-weight}, which is crucial to prove \cref{p:tree-sum}.
	The data structure is a~straightforward adaptation of \cite[Section 2]{EvenShiloach1981}, but we provide a full proof for completeness.
	
	\begin{lemma}\label{p:tree-sum-logn}
		There is a tree-sum data structure with the following running times for $n$ vertices: $\fO(n)$ for initialization, $\fO(1)$ for each \alg{tree-sum} and \alg{update-weight}, and $\fO(n \log n)$ in total for all \alg{cut-report} operations.
	\end{lemma}
	\begin{proof}
		Let $(F,w)$ be the maintained weighted forest.
		We keep an explicit representation of $F$ (with parent and child pointers) and a helper data structure $R$ on $F$ for \alg{root} queries (\cref{p:roots}). Each tree root $r$ stores the total weight of the respective tree, which we denote by $S[r]$. Each vertex additionally stores its current weight $w(v)$.
		
		Initialization clearly takes $\fO(n)$ time. To answer a $\alg{tree-sum}(v)$ query, we first compute $r \gets R.\alg{root}(v)$, and then return $S[r]$. For $\alg{update-weight}(v,x)$, first let $y = w(v)$ be the previous weight of $v$. Set $w(v) \gets x$, then again compute $r \gets R.\alg{root}(v)$ and set $S[r] \gets S[r] + x - y$. Each \alg{tree-sum} or \alg{update-weight} operation clearly takes constant time, as desired.
		
		We now implement $\alg{cut}(v)$ (and deal with \alg{cut-report} later). First determine $r \gets R.\alg{root}(v)$. Then perform $R.\alg{cut}(v)$ and remove the edge between $v$ and its parent from the explicit representation of $F$. Now $v$ and $r$ are both roots in the forest, and we need to (re)compute $S[v]$ and $S[r]$. Let $X$ be the previous value of $S[r]$.
		We compute the weight sums of each of the two trees $F_v$ and $F_r$ with a straightforward traversal, in $\fO(|F_v|)$, resp.\ $\fO(|F_r|)$ time.  However, we do this in parallel and stop when one of the two traversals is finished. In this way, we compute either $S[v]$ or $S[r]$ in $\fO(\min\{|F_v|,|F_r|\})$ time. If, say, we have $S[v]$, we can easily compute $S[r] \gets X - S[v]$, and vice versa.
		
		We now determine the total running time of all \alg{cut} operations. It is easy to see that it asymptotically follows the recursion
		\[ t(1) = 1;\quad t(n) \le \max_{k \in [n-1]} \left( \min\{ k, n-k \} + t(k) + t(n-k) \right), \]
		which is well-known to solve to $\fO(n \log n)$.
		
		Finally, \alg{cut-report} simply performs a \alg{cut} and then two appropriate \alg{tree-sum} operations. This takes $\fO(1)$ extra time per \alg{cut-report}, and since there are at most $n-1$ edges to delete, this is $\fO(n)$ extra work in total.
	\end{proof}
	
	\subsection{Iterating the algorithm}\label{sec:tree-sum-it}
	
	We now show how to build data structures with running times of roughly $\log\log n$, $\log\log\log n$, etc.\ per operation. The following reduction using cluster decompositions is the main necessary ingredient.
	
	\begin{lemma}\label{p:tree-sum-red}
		Suppose there is a tree-sum data structure $X$ with total running time $f(n',m')$ for $n'$ vertices and $m'$ operations.
		Then, there exists a tree-sum data structure with the following property. For each number of operations $m$ and number of vertices $n$, there exists an integer $k \in \N_+$ and integers $m_i \in \N_0$, $n_i \in \N_+$ for $i \in [k]$ with $\sum_{i=1}^k n_i \le n$, $\sum_{i=1}^k m_i \le m$, and $n_i \le \log n$ for all $i \in [k]$, such that the total running time on any instance with $n$ vertices and $m$ operations is at most
		\[ \fO(n+m) + \sum_{i=1}^k f(n_i,m_i). \]
	\end{lemma}
	
	The numbers $n_i$ in \cref{p:tree-sum-red} refer to the cluster sizes, and $m_i$ to the number of operations applied to a particular cluster. The lemma then essentially states that each operation is delegated to a single cluster, with overall linear overhead.
	
	In the remainder of this section, we prove \cref{p:tree-sum-red}.
	For simplicity, we assume the initial forest is a tree; otherwise, if the forest is disconnected, we can apply the data structure separately for each connected component, and use \cref{p:roots} to find the proper data structure to query or modify.
	Let $(T^0,w_0)$ be the initial weighted binary tree.
	We first use \cref{p:cluster-decomp} with the maximum cluster size $k = \log n$ to compute a cluster decomposition $\fC$ of $T^0$ with cluster tree $S$. We maintain the induced cluster forest $G$ as described by \cref{p:maintain-cluster-forest}.

	We also maintain a weight function $w'$ on the vertices in the cluster forest $G$, as follows. Suppose $(F,w)$ is the current weighted forest. Let $C$ be some cluster, and let $u = \ub(C) \in V(G)$. Then $w'(u)$ is defined as the sum of weights of all vertices in $C$ that are connected to $u$. On the other hand, if $v = \lb(C)$ exists, then $w'(v)$ is defined as the sum of weights of all vertices in $C$ that are connected to $v$, but not to $u = \ub(C)$. In particular, note that $w'(v) = 0$ as long as $v$ is connected to $u$.
	\begin{observation}\label{p:ts:w'-preserves}
		Let $v \in V(G)$ be a boundary vertex of some cluster, let $T$ be the tree in $F$ containing $v$, and let $U$ be the tree in $G$ containing $v$. Then $w(T) = w'(U)$.
	\end{observation}

	In the following, let $F$ denote the current forest, and let $G$ denote the cluster forest induced by $F$. Let $w$ denote the current weight function on $F$, and let $w'$ denote the weight function on $G$ defined as above. Our data structure maintains $(F,w)$ and $(G,w')$ explicitly (with child and parent pointers), and additionally:
	\begin{itemize}[nosep]
		\item A data structure $R$ for \alg{connected} queries in $F$ (\cref{p:roots}).
		\item An instance $D$ of the data structure from \cref{p:tree-sum-logn} on $(G,w')$.
		\item A \emph{cluster object} for each cluster $C$, which stores $\ub(C)$ and $\lb(C)$.
		\item For each cluster $C$, an instance $X_C$ of the given data structure $X$ (defined in \cref{p:cluster-decomp}) on the induced subforest $F[C]$, with the weight function $w$ restricted to~$C$.
		\item For each vertex, a pointer to the cluster object containing it.
	\end{itemize}
	
	Initialization works as follows. First compute $\fC$ and $G$ using \cref{p:cluster-decomp}, with $k = \log n$. Create the cluster objects and point each vertex towards its cluster. Then initialize $R$, $D$, and all instances $X_C$. Finally, compute $w'$ by summing the weights of each cluster and assigning this sum to the upper boundary vertex (recall that all clusters are still connected). We now describe the three operations.

	\begin{itemize}
		\item Consider an operation $\alg{tree-sum}(v)$. We first try to find a boundary vertex that is connected to $v$, by determining its cluster $C$, and then checking connectivity between $v$ and $\ub(C)$, resp.\ $\lb(C)$, using $R$. If $v$ is not connected to either, then the component of $F$ containing $v$ is entirely within $C$. Thus, we can return $X_C.\alg{tree-sum}(v)$. Otherwise, if $v$ is connected to some $u \in \{\ub(C),\lb(C)\}$, then we can return $D.\alg{tree-sum}(u)$ by \cref{p:ts:w'-preserves}.

		\item Consider $\alg{update-weight}(v,x)$. Say $y$ is the weight $w(v)$ before the operation. We again first determine the cluster $C$ containing $v$, and call $X_C.\alg{update-weight}(v,x)$. Further, the change to $w$ may affect $w'$ for either $\ub(C)$ or $\lb(C)$. We again check connectivity between $v$ and the two boundary vertices. If $v$ is connected to $\ub(C)$, then we increase $w'(\ub(C))$ by $x-y$. Otherwise, if $v$ is connected to $\lb(C)$, we increase $w'(\lb(C))$ by $x-y$. If neither is true, $w'$ does not change for any vertex.

		\item Consider $\alg{cut-report}(v)$. Let $u$ be the parent of $v$ (before the operation).
		Suppose first that $u$ and $v$ are in distinct clusters. Then $u$ is a lower boundary vertex and $v$ is an upper boundary vertex. Thus, we can simply call $D.\alg{cut-report}(v)$ and pass along the two returned values. No cluster is changed, so $w'$ and $X_C$ do not need to be touched.
		
		Second, suppose that $u$ and $v$ are both contained in a single cluster $C$. This still may remove an edge from $G$; we check this as in \cref{p:maintain-cluster-forest} and call $D.\alg{cut}$ if necessary.
		We then call $X_C.\alg{cut-report}(v)$, which returns the sums $x_u$ and $x_v$ of weights of vertices in $C$ connected to $u$, resp.\ $v$. Now $w'$ could have changed for $\ub(C)$ or $\lb(C)$, and we need to update it in $D$. Recalling the definition of $w'$, it is straightforward to determine $w'(\ub(C))$ and $w'(\lb(C))$ by checking connectivity between $\ub(C)$, $\lb(C)$, $u$, and $v$, and then using $x_u$ or $x_v$ if necessary.
		If $w'$ changes, we call $D.\alg{update-weight}$ appropriately.
		
		Finally, we need to return the results for $\alg{tree-sum}(u)$ and $\alg{tree-sum}(v)$. This is again straightforward to compute: If connected to $\ub(C)$ or $\lb(C)$, we can use $D.\alg{tree-sum}$; if not, we can use $x_u$ or $x_v$, respectively.
	\end{itemize}
	
	Note that we never call operations on auxiliary vertices in $X_C$, since by assumption we are never given an auxiliary vertex as parameter $v$.
	
	We now bound the running time. Initialization, aside from the instances $X_C$, takes time $\fO(n)$. Besides calls to $X_C$ and $D$, each operation performs a constant amount of work, for a total of $\fO(m)$ across $m$ operations.
	By \cref{p:tree-sum-logn}, the time spent in $D$ is $\fO(m + |V(G)| \log |V(G)|) \le \fO(m+n)$.
	
	Finally, consider the time spent in a single instance $X_C$. By definition, this is $f(n_i, m_i)$, where $m_i$ is the number of operations called on $X_C$, and $n_i = |C|$. It remains to show that the stated conditions on $m_i$ and $n_i$ are true. Clearly, the total size of all clusters is $n$, so $\sum n_i = n$. Also, each cluster has size at most $\log n$, so $n_i \le \log n$.
	Further, observe that each operation (in the data structure we just described) results in at most one call to an operation of $X_C$. Thus, we have $\sum m_i \le m$. This concludes the proof of \cref{p:tree-sum-red}.
	
	\begin{corollary}\label{p:tree-sum-param}
		There is a tree-sum data structure, parameterized with $t \in \N_+$, with total running time $\fO( t (m+n) + n \logit{t} n )$ for $m$ operations and $n$ vertices.
	\end{corollary}
	\begin{proof}
		Use \cref{p:tree-sum-logn} for the case $t = 1$. For $t \ge 2$, use \cref{p:tree-sum-red} and the data structure constructed for $t-1$. By induction, the running time of the latter is, up to an absolute multiplicative constant, bounded by $f(m',n') = t (m'+n') + n' \logit{t} n'$ for $m'$ operations and $n'$ vertices.
		The running time of the new data structure then is, up to an absolute multiplicative constant, at most
		\begin{align*}
			m+n + \sum_{i=1}^k (t-1) (m_i+n_i) + n_i \logit{t-1} n_i
			\le t(m+n) + n \logit{t} n,
		\end{align*}
		using that $\sum m_i = m$, $\sum n_i = n$, and $n_i \le \log n$.
	\end{proof}
	
	With $t = \log^* n$, \cref{p:tree-sum-param} implies \cref{p:tree-sum}.
	
	\section{Tree sizes in linear time}\label{sec:tree-size}
	
	In this section, we prove:
	
	\restateTreeSize*
	
	Recall the algorithm summary from the introduction; we now fill in the details.
	Let a \emph{tree-size data structure} be a data structure as in \cref{p:tree-0-1-sum}, which supports auxiliary vertices in the input.
	
	\begin{lemma}\label{p:tree-size-small}
		Fix a machine word size $b$ and some $\ell \le \tfrac{b}{3+\log b}$.
		There is a tree-size data structure that maintains a 0-1-weighted forest on up to $\ell$ vertices, with $\fO(\ell)$ initialization time and $\fO(1)$ time per operation.
		
		The data structure requires a global table, only depending on $\ell$, that can be precomputed in time $\fO( (5\ell)^\ell )$.
	\end{lemma}
	\begin{proof}
		We show how to handle forests with precisely $\ell$ vertices. If an input forest has fewer than $\ell$ vertices, we can simply add isolated vertices as necessary.
		
		Let $\fF$ be the set of all possible 0-1-weighted forests $(F,w)$ on $\ell$ vertices. We can always assume that $V(F) = [\ell]$.
		Encode each $(F,w)$ in a bitstring $\id(F,w)$ as follows: For each vertex $v \in [\ell]$, in order, encode its parent (zero if $v$ is a root) with $\lceil \log(\ell+1) \rceil$ bits. Then, encode its weight with one additional bit. The length of $\id(F,w)$ is $c \coloneq \ell (\lceil \log(\ell+1) \rceil + 1) \le b$, so it fits into a single machine word. Moreover, the number of bitstrings of length $c$ is at most $2^c \le (4\ell+4)^\ell$.
		
		We now describe the global table required by our data structure. Call it $A$. For each 0-1-weighted forest $(F,w)$ of size $\ell$, we store the following information in $A[\id(F,w)]$:
		\begin{itemize}
			\item A table $S$ of size $\ell$, containing the results of $\alg{tree-sum}(v)$ for each $v \in V(F)$.
			\item For $i \in \{0,1\}$, a table $U_i$ of size $\ell$. This table contains, for each $v \in V(F)$, the encoding $\id(F,w')$, where $w'$ is obtained from $w$ by setting $w'(v)$ to $i$. That is, let $w'(v) = i$ and $w'(u) = w(u)$ for all $u \neq v$.
			\item A table $C$ of size $\ell$ containing, for each vertex $v$, the encoding $\id(F',w)$, where $F'$ is obtained from $F$ by applying $\alg{cut}(v)$ (or $F' = F$, if $v$ is already a root in~$F$).
		\end{itemize}
		
		Computing these tables is easily done in $\fO(\ell)$ time (note that the encodings $\id(F,w')$ and $\id(F',w)$ can be computed from $\id(F,w)$). Since $A$ has size $2^c$, the total running time to compute $A$ is $\fO(2^c \cdot \ell) \le \fO( (5\ell)^\ell )$.
		
		Now suppose we are given a 0-1-weighted forest $(F,w)$. Our data structure consists of simply the bitstring $x \gets \id(F,w)$, which we construct in $\fO(\ell)$ time. Now suppose we call $D.\alg{tree-sum}(v)$, for $v \in [\ell]$. Then we find the entry $A[x]$, recover the table $S$, and return $S[v]$. For $D.\alg{update-weight}(v,i)$, we find $A[x]$, take the table $U_i$, and set $x \gets U_i[v]$. We can handle $D.\alg{cut}$ in the same way, using $C$. Each of these operations clearly takes constant time, as desired.
	\end{proof}
	
	We can now combine \cref{p:tree-size-small} with \cref{p:tree-sum-red} to prove \cref{p:tree-0-1-sum}.
	The only additional observation necessary is that the reduction of \cref{p:tree-sum-red} preserves the property that $w$ is a 0-1 weight function. Internally, the more general weight function $w'$ is used only for the data structure $D$ (which uses \cref{p:tree-sum-logn}), not for the data structures~$X_C$.
	
	\begin{proof}[Proof of \cref{p:tree-0-1-sum}.]
		Start by precomputing the global table of \cref{p:tree-size-small}, with parameter $\ell = \lfloor \log\log n \rfloor$. This takes $o(n)$ time. We now have a data structure with running time $f(m',n') \in \fO(m'+n')$ for $m'$ operations and $n'$ vertices, but only if $n' \le \log\log n$. Applying \cref{p:tree-sum-red} gives us a data structure with running time $\fO(m'+n')$ for all forests with $n' \le \log n$ vertices, and applying it the second time yields the claim.
	\end{proof}
	
	\section{Tree sums in optimal time}\label{sec:tree-sum-optimal}
	
	We now present an explicit data structure with unknown, but optimal running time. In this section, let a \emph{tree-sum} data structure support operations \alg{cut}, \alg{tree-sum}, and \alg{update-weight}; auxiliary vertices are again allowed. Unlike in \cref{sec:tree-sum}, we are not restricted to binary input trees, and do not require support for \alg{cut-report}.
	
	Let $D$ be a tree-sum data structure and let $F$ be a forest.
	Let $\Time(D,F,m)$ be the maximum number of additions and subtractions that $D$ performs when executing a legal sequence of $m \in \N_0$ operations with the initial forest $F$.
	The maximum is thus taken over the operation sequences and the initial weights.
	Let $\OPT(F,m)$ be the minimum $\Time(D,F,m)$ over all data structures $D$.
	Note that the quantity $\OPT(n,m)$ defined in the introduction is the maximum $\OPT(F,m)$ over all forests $F$ on $n$ vertices.\footnote{In the introduction, we did not allow auxiliary vertices; however, it is easy to see that marking a vertex as auxiliary cannot make the problem easier.} The main result of this section is:

	\begin{theorem}\label{p:tree-sum-univ-opt}
		There exists a single data structure with running time $\fO( \OPT(F,m) + m + n )$ for any forest $F$ on $n$ vertices and any sequence of $m$ operations. Note that the data structure does not know $F$, $n$, or $m$ upfront.
	\end{theorem}
	
	This means that our data structure performs asymptotically as good as even the data structures that have the number of operations $m$ and the initial forest $F$ hard-coded. In particular, \cref{p:tree-sum-univ-opt} implies \cref{p:tree-sum-opt}.

	As stated in the introduction, the main idea is to precompute an optimal data structure for very small trees.
	
	\begin{restatable}{lemma}{restateOptSmall}\label{p:opt-small}
		Given a forest $F$ with $n$ vertices, in time $2^{2^{n^{\fO(1)}}}$, we can precompute a tree-sum data structure $D$ that performs $m \in \N_0$ operations with the initial forest $F$ (and arbitrary initial weights) in time $\fO(\OPT(F,m) + m + n )$.
	\end{restatable}
	
	We will show \cref{p:opt-small} in \cref{sec:opt-small}. Our data structure (for large trees) works as follows. First, for each (ordered) binary tree $T$ on at most $k = \lfloor \log\log\log n \rfloor$ vertices, compute an optimal data structure $D_T$ for $T$ with \cref{p:opt-small}. As is well known, there are no more than $\fO(4^k) \le \fO(\log n)$ such trees.
	We can combine them into a single data structure $X$ with optimal running time (with linear overhead) for all binary trees of size at most $\log\log\log n$.\footnote{When given an initial binary tree $T$, the combined data structure $X$ can identify the respective $D_T$ in $\fO(|V(T)|)$ time, using a table containing each $D_T$.%
	}

	Now take the components $S^1, S^2, \dots, S^k$ of the input forest $F$. Binarize each of them using \cref{p:binarize}, obtaining binary trees $T^1, T^2, \dots, T^k$.
	As usual, we build the helper data structure of \cref{p:roots} to identify the affected tree when performing an operation on the overall forest.
	For each $T^i$, apply \cref{p:tree-sum-red} three times, using $X$ as the data structure for small trees.
	Note that \cref{p:tree-sum-red} requires $X$ to support \alg{cut-report}, which we simulate with one \alg{cut} and two \alg{tree-sum} operations.
	
	The running time for $m'$ operations on one of the trees $T^i$ is as follows, for some $k \in \N_+$, some partition of $T^i$ into (cluster) subtrees $T^{i,1}, T^{i,2}, \dots, T^{i,k}$, and some $m_1, m_2, \dots, m_k \in \N_0$ with $\sum_{j=1}^k m_j \leq m'$:
	\[ \fO(|V(T^i)|+m') + \sum_{j=1}^k \OPT(T^{i,j},3m_j). \]
	The factor 3 in $\OPT(T^{i,j},3m_j)$ is due to the fact that each \alg{cut-report} operation results in three actual operations on $X$.
	
	We now need the following lemma, which we will show in a~moment in \cref{sec:opt-superadd}.
	
	\begin{restatable}{lemma}{restateOptSuperadd}\label{p:opt-superadd}
		Let $F$ be a forest, and let $T_1, T_2, \dots, T_k$ be the trees induced by a partition of $V(F)$ into connected subsets. Further let $m_1, m_2, \dots, m_k \in \N_0$. Then:
		\[ \sum_{i=1}^k \OPT(T_i,3m_i) \le \fO( \OPT(F,m) + n). \]
	\end{restatable}

	This bounds the running time for $m'$ operations on $T^i$ by $\fO(|V(T^i)| + m' + \OPT(T^i,m') )$. Now observe that $\OPT(S^i,m') = \OPT(T^i,m')$, since by \cref{p:binarize}, all operation sequences return exactly the same results on $S^i$ and $T^i$. Applying \cref{p:opt-superadd} again on the components $S^1, S^2, \dots, S^k$ yields the running time $\fO(n+m+\OPT(F,m))$ for $m$ operations on $F$. Thus, we have shown \cref{p:tree-sum-opt}.
	
	\subsection{Weak super-additivity of \texorpdfstring{$\OPT$}{OPT}}\label{sec:opt-superadd}
	
	In this section, we prove \cref{p:opt-superadd}.
	We need the following definitions.
	An \emph{initial-zero tree-sum data structure} is a tree-sum data structure that only accepts initial forests where all weights are zero. (Thus, weights need to be changed before obtaining anything useful from \alg{tree-sum} queries).
	Define $\OPTZ$ in the same way as $\OPT$, but for initial-zero tree-sum data structures.
	It turns out that $\OPTZ$ is easier to work with than $\OPT$. Indeed, we can show that $\OPTZ$ is \emph{super-additive}.

	In the following statement, we say that $U \subseteq V(F)$ is \emph{convex} if for each $u, v \in U$, if $u$ and $v$ are connected in $F$, then they are also connected in $F[U]$.
	In other words, $U$ includes from each tree $T$ of $F$ either a~connected subgraph of $T$, or nothing at all.
	In the initial-zero setting, this in particular implies that any sequence of operations in an~instance $(F, w)$ that is legal in $(F[U], w|_U)$ yields the same results in both instances.
	
	\begin{lemma}\label{p:optz-superadd}
		Let $F$ be a forest, $U_1, U_2$ be a partition of $V(F)$ into convex subsets, and let $m_1, m_2 \in \N_0$.
		Then $\OPTZ(F[U_1],m_1) + \OPTZ(F[U_2], m_2) \le \OPTZ(F, m_1+m_2)$.
	\end{lemma}
	\begin{proof}
		Let $F_1 = F[U_1]$, let $F_2 = F[U_2]$, and let $m = m_1+m_2$. Let $D$ be an optimal initial-zero tree-sum data structure for $F, m$, that is, the maximum number of additions and subtractions $D$ performs when handling any $m$ operations with the initial forest $F$ and zero weights is precisely $\OPTZ(F,m)$. We now give two data structures $D_1$ and $D_2$ for the forests $F_1$ and $F_2$, respectively, and then show that they perform at most $\OPTZ(F, m)$ additions and subtractions in total when presented with $m_1$ and $m_2$ operations, respectively.
		
		First, $D_1$ is the same as $D$. In particular, $D_1$ behaves as if initialized with the original forest~$F$, although it only accepts $\alg{cut}(v)$, $\alg{update-weight}(v, \cdot)$ or $\alg{tree-sum}(v)$ if $v \in U_1$. %
		Now take an operation sequence $X_1^*$ of length $m_1$ on $F_1$ such that the number of additions and subtractions $D$ performs to execute $X_1^*$ is \emph{maximized}.
		Run $D$ with $X_1^*$, and capture the final state of the data structure. Hard-code this state into the data structure $D_2$, except set the weights of all vertices in $F_1$ and all group elements stored by the data structure to zero. When presented with an operation sequence $X_2$ on $F_2$ (one operation at the time), $D_2$ continues the execution of $D$ from that final state.\footnote{This formulation assumes that data structures have access to the zero element of the weight group. However, note that we can easily simulate additions with zero instead of actually storing zeroes.}
		
		The correctness of $D_1$ is easy to see: Suppose $X_1$ is a~sequence of $m_1$ operations on $F_1$.
		Consider a $\alg{tree-sum}(v)$ query in $X_1$ (possibly after some \alg{cut}s).
		In $D_1$, the query returns the sum of weights of some set $U \subseteq V(F_1)$.
		In $D$, it returns the sum of weights of some superset $U' \supseteq U$.
		However, since $U$ and $U'$ induce connected subgraphs of $F$, and $V(F_1)$ is convex, we have $U' \setminus U \subseteq V(F_2)$.
		Moreover, weights in $V(F_2)$ are zero at the start and never changed by $X_1$. Thus, the query returns the same result in both cases.
		
		We can show correctness of $D_2$ in a similar, but slightly more complicated way. Consider a $\alg{tree-sum}(v)$ query from $X_2$ applied to $D$, assuming $X_1^*$ and all previous queries in $X_2$ have been applied already.
		The return value $x$ of this query is computed by a series of additions and subtractions involving (previous or current) weights from $F_1$ and~$F_2$.
		By ignoring all such additions and subtractions performed during the execution of $X_1^*$, we get that $x$ is a linear combination of (1) weights from $F_1$ and $F_2$, all read during the execution of $X_2$, and (2) values stored in the data structure directly after execution of $X_1$.
		Now recall that $D_2$ sets all those stored values to zero, and also sets all weights in $F_1$ to zero. Thus, what $D_2$ returns is the reduced linear combination with only the weights from $F_2$. As above, by convexity, this is precisely the correct result.
		
		Having shown correctness, we now argue the stated running times. Let $X_1$ and $X_2$ be some operation sequences on $F_1$ and $F_2$, respectively.
		Recall that $\Time(D',F',m')$ is the maximum number of additions and subtractions that a data structure $D'$ performs when executing $m' \in \N_0$ operations with an initial forest $F'$.
		For convenience, let also $\Time(D',F',X')$ be the number of additions and subtractions that $D'$ performs when executing a specific sequence $X'$ of operations on $F'$. Note that $\Time(D',F',m')$ is the maximum of $\Time(D',F',X')$ over all sequences $X'$ of length $m'$.

		By construction, we have $\Time(D_1, F_1, X_1) = \Time(D, F, X_1)$, and $\Time(D, F, X_1^*) + \Time(D_2, F_2, X_2) = \Time(D, F, X_1^* + X_2)$, and further $\Time(D, F, X_1) \le \Time(D, F, X_1^*)$.
		These three inequalities imply that
		\[\Time(D_1,F_1,X_1)+\Time(D_2,F_2,X_2) \le \Time(D,F,X_1^*+X_2) \le \OPTZ(F,m).\]
		
		Since this is true for any pair of sequences $X_1$ and $X_2$, we have
		\[\OPTZ(F_1,m_1)+\OPTZ(F_2,m_2) \le \OPTZ(F,m),\]
		as desired.
	\end{proof}
	
	We will now generalize \cref{p:optz-superadd} to normal tree-sum data structures without the initial-zero assumption. We need a technical lemma about the relation between $\OPT$ and $\OPTZ$.
	
	\begin{restatable}{lemma}{restateRelationOptOptZ}\label{p:relation-opt-optz}
		We have $\OPT(F,m) \le \OPTZ(F,5m) + 2n$ for each forest $F$ on $n$ vertices and $m \in \N_0$. %
	\end{restatable}
	
	The proof of \cref{p:relation-opt-optz} requires a forest partition similar to a cluster decomposition (see \cref{p:cluster-decomp}). We cannot use the latter, since it requires binary trees, and we need the lemma to work for arbitrary trees.
	Fortunately, not all properties of cluster decompositions are needed, so we can use the following simplification.
	
	\newcommand{\fP}{\mathcal{P}}
	
	\begin{lemma}\label{p:alt-decomp}
		Let $F$ be a (not necessarily binary) forest and let $k \in \N_+$. Then there exists a partition $\fP$ of $V(F)$ such that each $P \in \fP$ induces a (connected) subtree of $F$, and, for each $P \in \fP$,
		\begin{enumerate}[nosep, label = (\alph*)]
			\item either the root of $F[P]$ is a root of $F$ and $|P| < k$,
			\item or $|P| \ge k$ and removing the root of $F[P]$ splits $F[P]$ into subtrees of size less than $k$.
		\end{enumerate}
	\end{lemma}
	\begin{proof}
		We can use a simplified version of the algorithm in \cref{p:cluster-decomp}: Essentially, grow sets $P$ bottom-up, and finish a set when it grows to at size $k$ or larger.
		
		Analogously to \cref{p:cluster-decomp}, we define a recursive algorithm $\alg{construct}'(v)$ that will be applied to every root $v$ of $F$ and returns the part $P \in \fP$ containing $v$.
		It works as follows. If $v$ has children, apply $\alg{construct}'$ to each of them. Let $\fC$ be the set of returned parts that have size less than $k$. Destroy these parts and return $P = \{v\} \cup \bigcup \fC$.
		In the end, let $\fP$ be the set of all parts that were created, but not destroyed.
		
		It is easy to see that every vertex is contained in some part in $\fP$, and that each part satisfies one of the two requirements of the lemma.
	\end{proof}

	\begin{proof}[Proof of \cref{p:relation-opt-optz}]
		First, observe that an initial-zero tree-sum data structure can simulate a (normal) tree-sum data structure by using $\alg{update-weight}$ at the start to set the weight of each non-auxiliary vertex (recall that auxiliary vertices have weight zero). This requires $n$ operations, so we have $\OPT(F,m) \le \OPTZ(F, m+n)$. If $n \le 4m$, we are done.
		
		Now suppose that conversely $m \le \tfrac14 n$. Take an initial-zero tree-sum data structure $D_0$ that optimally handles $5m$ operations on $F$. We will present a tree-sum data structure $D$ for $F$ that handles $m$ operations with at most $\fO(n)$ additional operations.
		
		Given is $F$ and an initial weight function $w$ on $F$.
		Our data structure $D$ first computes a decomposition $\fP$ as defined in \cref{p:alt-decomp} with $k = \lceil \tfrac n m \rceil$.
		Note that we do not care about running time, only number of additions and subtractions (of which we have done none so far).

		Call a part \emph{special} if it induces a component of the current forest and has size smaller than~$k$. Note that each \alg{cut} splits some part into two new parts. If the original part was special, the two new parts must also both be special. Otherwise, up to one new part may be special.
		
		We claim that initially there are at most $2 \tfrac n k$ non-special parts. Indeed, there clearly are at most $\tfrac n k$ parts of size at least~$k$.
		Further, each component that is not induced by a special part contains at most one part of size less than $k$ (the one containing the root), and at least one other part (which must have size at least $k$ by \cref{p:alt-decomp}). Thus, the number of such small non-special parts is no more than $\tfrac n k$, so overall we have no more than $2 \tfrac n k$ non-special parts.
		
		For now, let us assume that $F$ contains no auxiliary vertices.
		Initialization works as follows.
		For each non-special part $P$, we compute the sum of its weights~$x$. Then, we call $D_0.\alg{update-weight}(r, x)$, where $r$ is the root of $F[P]$. This requires fewer than $n$ additions and subtractions, and at most $2 \tfrac n k$ \alg{update-weight} operations on~$D_0$. Observe that at this point of time, the sum of weights of each non-special component is correct in $D_0$; we will maintain this invariant throughout.
		
		Now consider an operation $\alg{tree-sum}(v)$. If $v$ is contained in a special part, we compute the sum of the respective special component by brute force, consulting the weight function $w$. Since special components have size smaller than $k$, this takes fewer than $k-1$ operations.
		If $v$ is contained in a non-special part, we delegate the query to $D_0$.
		The result will be correct by our invariant.
		
		Second, consider an operation $\alg{update-weight}(v,x)$. We update $w$ accordingly. If $v$ is contained in a non-special part $P$, we call $D_0.\alg{update-weight}(r,y)$, where $r$ is the root of $F[P]$, and $y$ is the old weight of $r$ in $D_0$, plus the difference between $x$ and the previous weight $w(v)$. This requires up to one addition and one subtraction, and maintains the invariant.
		
		Finally, consider an operation $\alg{cut}(v)$. If $v$ and its parent $u$ are contained in the same non-special part $P$, we split $P$ into the two parts $P_1$ and $P_2$ that stay connected after removing the edge $\{u,v\}$.
		One or both of the new parts may be non-special, so we need to store their weight sums in their respective roots.
		By \cref{p:alt-decomp}, when splitting any part in this way (even repeatedly), at least one of the newly created parts has size at most $k-1$. Thus, we can compute the weight sum of one of $P_1$ and $P_2$ with at most $k-2$ additions, and compute the weight sum of the other part with a single subtraction.
		Then, we update the root weights in $D_0$ with two \alg{update-weight} operations, and finally perform $\alg{cut}(v)$ in~$D_0$, for a total of three operations in $D_0$.
		
		For initialization and $m$ operations, our data structure performs at most $2 \tfrac n k + 3m \le 5m$ operations in $D_0$, and additionally performs at most $n + m \cdot (k-1) \le 2n$ additions and subtractions, as desired.
		
		It remains to show how to support auxiliary vertices in $F$. The algorithm described above may illegally access an auxiliary vertex when performing $D_0.\alg{update-weight}(r,\cdot)$ on the root $r$ of a part; this can happen during initialization or as the result of an \alg{update-weight} or \alg{cut} operation.
		To avoid this, we define the \emph{representative} of a part as an arbitrary, but fixed non-auxiliary vertex. Assuming such a vertex exists, we can simply use the representative instead of the root. If no such vertex exists, all vertices in the given part have weight zero by definition. In that case, we can skip the call $D_0.\alg{update-weight}(r,x)$ during initialization ($x = 0$ anyway). Moreover, no \alg{update-weight} or \alg{cut} operation can be applied to any vertex in that part, so no $D_0.\alg{update-weight}$ will be ever necessary.
	\end{proof}

	We need one more property of $\OPT$, namely that several instances of a~data structure processing $m$ operations on a~forest $F$ can be combined into a single instance processing significantly more than $m$ operations on $F$.

	\begin{lemma}\label{p:opt-num-ops-nonbinary}
		Let $F$ be a forest and $m \in \N_+$. Then $\OPT(F,\left\lceil\tfrac32 m \right\rceil) \le 3 \cdot \OPT(F,m)$.
	\end{lemma}
	\begin{proof}
		Suppose we have an optimal data structure $D_m$ for $m$ operations on $F$.
		We handle a sequence $X$ of at most $\left\lceil\tfrac32 m\right\rceil$ operations as follows.
		Initialize an instance $I$ of $D_m$ with the initial forest and weight assignments, and perform the first $m$ operations of $X$ on $I$.
		Let $(F',w')$ be the weighted forest obtained after performing the first $m$ operations.
		Let $E$ be the set of edges of $F$ that were cut by the first $m$ operations of $X$, so $F' = F - E$.
		Let also $\mathcal{C}$ be the set of connected components of $F'$.
		For any subset $\mathcal{C}' \subseteq \mathcal{C}$, define the \emph{cost} of $\mathcal{C}'$, denoted $c(\mathcal{C}')$, as the number of edges in $E$ that need to be cut to separate all components in $\mathcal{C}'$ from each other.
		Naturally, $c(\mathcal{C}) = |E| \leq m$.
		In general, $c(\mathcal{C}')$ equals the sum, over all trees $T$ of $F$ containing at least one component from $\mathcal{C}'$, of the number of components of $\mathcal{C}'$ contained in $T$ minus one.

		We will now prove that $\mathcal{C}$ can be partitioned into two families $\mathcal{C}_1, \mathcal{C}_2$ such that $c(\mathcal{C}_1), c(\mathcal{C}_2) \le \left\lfloor \tfrac12 m \right\rfloor$.
		Indeed, consider any tree $T$ of $F$ containing $k$ components from $\mathcal{C}$, say $A_1, \ldots, A_k$.
		The contribution of $T$ to $c(\mathcal{C})$ is $k-1$.
		Assigning the components $A_1, \ldots, A_k$ to $\mathcal{C}_1$ and $\mathcal{C}_2$ evenly ensures that each $\mathcal{C}_i$ receives at most $\left\lceil \tfrac{k}{2} \right\rceil$ components from $T$, and thus the contribution of $T$ to $c(\mathcal{C}_i)$ is at most $\left\lceil \tfrac{k}{2} \right\rceil - 1 \le \tfrac{k-1}{2}$.
		Therefore, $c(\mathcal{C}_1), c(\mathcal{C}_2) \le \tfrac12 c(\mathcal{C}) \le \tfrac12 m$, so in particular $c(\mathcal{C}_1), c(\mathcal{C}_2) \le \left\lfloor \tfrac12 m \right\rfloor$.

		We now initialize two more instances $I_1$ and $I_2$ of $D_m$. The instance $I_i$ will receive the initial forest $(F,w'_i)$ with the following initial weights:
		\[
			w'_i(v) = \begin{cases}
				w'(v) & \text{if $v$ is contained in a component from $\mathcal{C}_i$,}\\
				0 & \text{otherwise.}
			\end{cases}
		\]
		Then, the instance $I_i$ immediately receives a~sequence of \alg{cut}s separating all components in $\mathcal{C}_i$ from each other; by our previous discussion, this requires at most $\left\lfloor \tfrac12 m \right\rfloor$ operations in $I_i$.
		At this point, observe that the \alg{tree-sum} of any vertex $v$ in $F'$ can be computed by adding the results of $\alg{tree-sum}(v)$ in both $I_1$ and $I_2$.
		Now perform the remaining $\left\lceil \tfrac12 m \right\rceil$ operations from $X$ as follows:
		\begin{itemize}
			\item \alg{update-weight} operations are delegated to the instance $I_i$ responsible for the component containing the affected vertex;
			\item \alg{cut} operations are forwarded to both $I_1$ and $I_2$; and
			\item \alg{tree-sum} operations are performed by querying both $I_1$ and $I_2$ and returning the sum of the results.
		\end{itemize}
		Clearly, this is correct, and each instance $I_i$ performs in total at most $\left\lfloor \tfrac12m \right\rfloor + \left\lceil \tfrac12m \right\rceil = m$ operations.
		This yields the required bound.
	\end{proof}
	
	Applying \cref{p:opt-num-ops-nonbinary} a constant number of times yields $\OPT(F, c \cdot m) \le \fO( \OPT(F,m) )$ for every fixed constant $c \ge 1$. With this, we are finally ready to prove:
	
	\restateOptSuperadd*
	\begin{proof}
		We have
		\begin{align*}
			\sum_{i=1}^k \OPT(T_i, 3m_i) & \le \sum_{i=1}^k \OPTZ(T_i,15m_i) + \fO(n) & \text{(\cref{p:relation-opt-optz})}\\
			& \le \OPTZ(F,15m) + \fO(n) & \text{(\cref{p:optz-superadd})}\\
			& \le \OPT(F,15m) + \fO(n) & \text{(obvious)}\\
			& \le \fO(\OPT(F,m) + n). & \text{(\cref{p:opt-num-ops-nonbinary})}
		\end{align*}
	\end{proof}
	Observe that, unsurprisingly, the constant $3$ in \cref{p:opt-superadd} can be replaced by an arbitrary constant.
	We believe that the function $\OPT$ is properly super-additive.
	That is, we make the following conjecture:
	\begin{conjecture}
		Let $F$ be a forest, $U_1, U_2$ be a partition of $V(F)$ into convex subsets, and let $m_1, m_2 \in \N_0$.
		Then $\OPT(F[U_1],m_1) + \OPT(F[U_2], m_2) \le \OPT(F, m_1+m_2)$.
	\end{conjecture}
	
	\section{Hardness of weighted subtree-sum}
	\label{sec:sumtree-sum-weights-with-cut-lb}
	
	In this section we prove \cref{p:sumtree-sum-weights-with-cut-lb}, restated here for convenience.
	\restateSubtreeSumHard*

	We remark that the lower bound holds even in the more general \emph{cell-probe} model of computation, where the time complexity of an~operation is only measured in the number of memory accesses (probes) it makes, regardless of the computation performed between probes.
	Moreover, we allow for Las Vegas randomization. Auxiliary vertices are not allowed in this section.
	
	In our hardness proof, we invoke the P{\u{a}}tra{\c{s}}cu and Demaine's lower bound for the \alg{partial-sum} problem~\cite{PatrascuDemaine2006}.
	Recall that in this problem, we maintain an~array $A$ of length $n$ with values from $0$ to $U - 1$ for some $U \geq n^{\Omega(1)}$, supporting operations of the form $\alg{update}(i, x)$ (set $A[i] \gets x$) and $\alg{partial-sum}(p)$ (return $\sum_{i=1}^p A[i]$).
	We will say that a~sequence of operations performed on $A$ is \emph{epoch-based} if all operations can be partitioned into some $E$ epochs, such that within each epoch, each element of $A$ is \alg{update}d exactly once and the sum of each prefix is queried exactly once.
	We will use the following version of the lower bound:
	\begin{restatable}{theorem}{restateAdaptedPatrascuDemaine}
		\label{t:adapted-prefix-sum-lb}
		Consider any cell-probe data structure for the \alg{partial-sum} problem on an~array $A$ of length $n$ and values from $0$ to $U \geq n^{\Omega(1)}$ that may use Las Vegas randomization.
		Let $E \leq n^{\fO(1)}$.
		The data structure requires $\Omega(En \log n)$ time to process $\Theta(En)$ operations, even if we assume that the sequence of operations is epoch-based and comprises $E$ epochs.
	\end{restatable}
	We remark that \cref{t:adapted-prefix-sum-lb} is not explicitly stated in \cite{PatrascuDemaine2006}, but it can be derived from their techniques; we provide a proof in \cref{sec:adapted-prefix-sum-lb}.
	
	Towards \cref{p:sumtree-sum-weights-with-cut-lb}, assume we have a~data structure $X$ supporting \alg{subtree-sum} and \alg{cut}.
	Using $X$, we will construct a data structure $Y$ for \alg{partial-sum} from \cref{t:adapted-prefix-sum-lb}, as follows.
	Suppose $Y$ is initialized with an~array $A$ of length $n'$ containing integers from $0$ to $n' - 1$.
	Let $n = n'(9n' + 2)$.
	We construct a weighted forest $F$ comprising two trees: $T_+$ on $n_+ = n'(8n' + 1)$ vertices and $T_-$ on $n_- = n'(n' + 1)$ vertices.
	$T_+$ contains a~path $v^+_1 v^+_2 \ldots v^+_{n'}$ of $n'$ vertices, where each $v^+_i$ has weight $0$.
	For each $i \in [n']$, we create $8n'$ additional vertices, denoted $u_{i,1}, u_{i,2}, \ldots, u_{i,8n'}$, where $u_{i,j}$ has weight $j$.
	We connect each $u_{i,j}$ and $v^+_i$ with an edge.
	$T_-$ is constructed similarly, with a~path $v^-_1 v^-_2 \ldots v^-_{n'}$ of $n'$ vertices of weight $0$, and for each $i \in [n']$, $n'$ additional vertices $w_{i,1}, w_{i,2}, \ldots, w_{i,n'}$ of weight $1$ connected to $v^-_i$.
	We root $T_+$ at $v^+_{n'}$ and $T_-$ at $v^-_{n'}$, so that $\alg{subtree-sum}(v^+_i)$ (respectively, $\alg{subtree-sum}(v^-_i)$) returns the sum of weights of all vertices $u_{i',j}$ (resp., $w_{i',j}$) with $i' \le i$.
	We remark that $F$ contains no auxiliary vertices.
	The data structure $X$ is initialized with $F$ and the given weights.
	
	We will now show how to use $X$ to process an~epoch-based sequence of operations on $A$ comprising $E = n'$ epochs, with a total of $\Theta((n')^2) = \Theta(n)$ operations.
	Suppose an~\alg{update} is performed on $A$ during the $p$th epoch, replacing $A[i]$ with a~new value from range $[0, n' - 1]$ and thus increasing $A[i]$ by some $\Delta \in [-n' + 1, n' - 1]$.
	Then $X$ performs three \alg{cut}s: it cuts $u_{i,j_1}$ and $u_{i,j_2}$ from $T_+$, where $j_1$ and $j_2$ are chosen such that $j_1 + j_2 = 7n' - \Delta$ and both $u_{i,j_1}$ and $u_{i,j_2}$ are still connected to $v^+_i$, and it cuts $w_{i,p}$ from $T_-$.
	Later we will show that such $j_1$ and $j_2$ always exist and can be found by a~randomized algorithm in expected constant time.
	Observe that after these cuts, the value of the expression
	\[
	\alg{subtree-sum}(v^+_k) - 7n' \cdot \alg{subtree-sum}(v^-_k)
	\]
	remains unchanged for $k < i$, and increases by $\Delta$ for $k \ge i$.
	Thus, given suitable preprocessing, we can compute any partial sum of $A$ in constant time using two \alg{subtree-sum} queries -- one in each of $T_+$ and $T_-$.
	The details follow.
	
	At the time of the initialization, we construct a table $B$ of size $n'$, so that
	\begin{equation}
		\sum_{i=1}^k A[i] = \alg{subtree-sum}(v^+_k) - 7n' \cdot \alg{subtree-sum}(v^-_k) + B[k]
		\quad \text{for all } k \in [n'].
		\label{eq:prefix-sum-via-subtree-sum}
	\end{equation}
	Note that initially, $\alg{subtree-sum}(v^+_k) = 4kn'(8n' + 1)$ and $\alg{subtree-sum}(v^-_k) = k n'$ for all $k \in [n']$, so $B$ can be computed in $\fO(n')$ time.
	The invariant \eqref{eq:prefix-sum-via-subtree-sum} will be maintained throughout the execution of the data structure.
	We will also preserve the following invariant: after $p$ \alg{update}s to the value of $A[i]$, exactly $2p$ vertices $u_{i,j}$ have been cut from $T_+$, and exactly vertices $w_{i,p'}$ for $p' \le p$ have been cut from $T_-$.
	
	Consider now an operation during the $p$th epoch.
	A~\alg{partial-sum} query for some $k \in [n']$ can be answered by computing the right-hand side of \eqref{eq:prefix-sum-via-subtree-sum}, which requires two \alg{subtree-sum} queries and a constant amount of additional work.
	Now consider an~\alg{update} changing $A[i]$ by $\Delta \in [-n' + 1, n' - 1]$, and define $s = 7n' - \Delta \in [6n' + 1, 8n' - 1]$.
	
	\begin{lemma}
		There exist $j_1, j_2 \in [8n']$ such that $j_1 + j_2 = s$ and both $u_{i,j_1}$ and $u_{i,j_2}$ are still connected to $v^+_i$.
		Furthermore, such $j_1$ and $j_2$ can be found by a~Las Vegas randomized algorithm in expected constant time.
	\end{lemma}
	\begin{proof}
		Consider the $3n'$ pairs of integers $(j, s - j)$ for $j \in [3n']$.
		Since $s \in [6n' + 1, 8n' - 1]$, all $6n'$ integers in these pairs are distinct and belong to $[8n']$.
		Since $p < n'$, by the invariant at most $2(p-1) < 2n'$ of these integers correspond to vertices $u_{i,j}$ that have already been cut from $T_+$.
		Therefore, at least $n'$ pairs $(j, s - j)$ remain such that both $u_{i,j}$ and $u_{i,s - j}$ are still connected to $v^+_i$.
		Thus, the sought pair $(j_1, j_2)$ exists, and it can be found by randomly sampling pairs as above until a valid one is found, which takes expected constant time.
	\end{proof}
	
	After finding such $j_1$ and $j_2$, $X$ performs $\alg{cut}(u_{i,j_1})$, $\alg{cut}(u_{i,j_2})$, and $\alg{cut}(w_{i,p})$.
	This preserves all the invariants, as argued above, completing the description of the data structure $Y$.
	However, by \cref{t:adapted-prefix-sum-lb}, $Y$ requires $\Omega(E n' \log n') = \Omega(n \log n)$ time to process $E = n'$ epochs of operations.
	$Y$ manages to perform the entire sequence of operations via $\Theta((n')^2) = \Theta(n)$ operations on $X$ and a~total of $\fO((n')^2) = \fO(n)$ additional work (for maintaining $B$ and finding $j_1, j_2$).
	Thus, $X$ must require $\Omega(n \log n)$ time to process the $\Theta(n)$ operations, proving \cref{p:sumtree-sum-weights-with-cut-lb}.

	We also observe that this proof can be rewritten to also apply in the group model setting:
	
	\begin{corollary}
		Each data structure that maintains a weighted forest under $n$ vertices under the operations \alg{subtree-sum} and \alg{cut} requires $\Omega( n \log n)$ time for $\Theta(n)$ operations in the group model of computation.
	\end{corollary}
	\begin{proof}
		Let $B = o(n)$ be such that the weights of the vertices in the forest $F$ above are strictly smaller than $B$, and let $D$ be the considered data structure in the group model of computation.
		Since $D$ works correctly in any commutative group, it can report the correct subtree sums in the group $\mathbb{Z}_{nB}$ of integers modulo $nB$; these are equal in value to the subtree sums in $\mathbb{Z}$, as all subtree sums are strictly smaller than $nB$.
		Meanwhile, all operations in $\mathbb{Z}_{nB}$ can be simulated by a Word RAM machine in constant time; therefore there exists a~data structure $D'$ in the Word RAM model of computation with the same complexity guarantees that maintains $F$ under \alg{subtree-sum} and \alg{cut}.
		\Cref{p:sumtree-sum-weights-with-cut-lb} applied to $D'$ finishes the proof.
	\end{proof}

	\section{Subtree sizes}\label{sec:subtree-size}
	
	In this section, we consider the \alg{subtree-sum} problem on \emph{unweighted} and \emph{0-1-weighted} forests, where \alg{update-weight} is not available and $\alg{subtree-sum}(v)$ returns the number of vertices (resp., the number of vertices with weight $1$) in the subtree rooted at $v$. Our algorithms (\cref{sec:subtree-size-ub}) support auxiliary vertices. We determine the complexity of the problem to be $\Theta(n \frac{\log n}{\log\log n})$ for $\Theta(n)$ operations. %

	\subsection{Lower bound}\label{sec:subtree-size-lb}
	
	We prove the lower bound (\cref{p:subtree-size-lb}) below:
	
	\restateSubtreeSizeLB*
	
	Similarly to \cref{p:sumtree-sum-weights-with-cut-lb}, we reduce from the \alg{partial-sum} problem; however, we find it more convenient to invoke a~result from an~earlier work by Fredman and Saks~\cite{FredmanSaks1989} that concerns the \alg{partial-sum-parity} problem.
	In \alg{partial-sum-parity}, we are given an~array $A$ of length $n$ with values from $\{0,1\}$, and we should support two operations: \alg{update} (flipping the value of $A[i]$) and \alg{partial-sum-parity} (returning the parity of the sum of the first $i$ values of $A$).
	
	\begin{theorem}[{\cite[Theorem 3]{FredmanSaks1989}}]\label{t:fredman-saks-lb}
		Consider any cell-probe data structure for the problem of \alg{partial-sum-parity} on an~array $A$ of length $n$ that may use amortization and Las Vegas randomization.
		The data structure requires $\Omega(n \tfrac{\log n}{\log\log n})$ time to process $\Theta(n)$ operations, even if each element of the array is \alg{update}d at most once.
	\end{theorem}
	
	\begin{proof}[Proof of \cref{p:subtree-size-lb}]
		Let $X$ be a~data structure for \alg{subtree-sum} on unweighted forests.
		Using $X$, we will construct a~data structure $Y$ for \alg{partial-sum-parity} on an~array $A$ of length $n'$.
		On initialization of $Y$ with an~array $A$, we construct a~tree $T$ of size $n \leq 3n'$ as follows.
		Let $v_1v_2 \ldots v_{n'}$ be a~path of $n'$ vertices.
		For each $i \in [n']$, if $A[i] = 0$, we create one additional vertex $u_{i,1}$ connected to $v_i$; if $A[i] = 1$, we create two additional vertices $u_{i,1}$ and $u_{i,2}$ connected to $v_i$.
		We root $T$ at $v_{n'}$.
		This construction ensures the following invariant: For each $i \in [n']$, the parities of $\alg{subtree-sum}(v_i)$ and $\sum_{j=1}^i A[j]$ are the same.
		The data structure $X$ is initialized with $T$.
		Then, to process an~\alg{update} of $A$ flipping the parity of $A[i]$, we simply invoke $\alg{cut}(u_{i,1})$; this is legal (since each element of $A$ is \alg{update}d at most once) and preserves the invariant.
		Therefore, a~\alg{partial-sum-parity} query for some $k \in [n']$ can be answered by simply invoking $\alg{subtree-sum}(v_k)$ and returning its parity.
		
		By \cref{t:fredman-saks-lb}, $Y$ requires $\Omega(n' \tfrac{\log n'}{\log\log n'})$ time to process $\Theta(n')$ operations.
		On the other hand, our implementation of $Y$ performs $\Theta(n')$ operations on $X$ and a~total of $\fO(n')$ additional work.
		Since $n' = \Theta(n)$, this finishes the proof.
	\end{proof}
	
	\subsection{Upper bound}\label{sec:subtree-size-ub}
	
	We start with a data structure for forests that are very small relative to the size of a machine word, which we denote by $b$.
	This is somewhat similar to \cref{p:tree-size-small}. However, here we actually need to maintain small \emph{weighted} forests, which means that we need a different approach.
	The weights are assumed to fit in one word each, i.e., to be an integer in $[0,2^b-1]$. We only allow a very limited way of changing weights: the operation $\alg{decrement-weight}(v)$, which decreases the weight $w(v)$ of $v$ by one (only allowed if $w(v) \ge 1$).
	Note that this data structure is merely used as a subroutine; in the end, we give a data structure for \emph{unweighted} forests of size $n = \Theta(2^b)$. %
	The proof can be thought of as an adaptation of Dietz's data structure for \emph{Subset Rank}~\cite{Dietz1989} -- essentially, partial sums of 0-1-arrays -- to forests.
	
	\begin{lemma}\label{p:subtree-size-small}
		Fix a word size $b$, and let $k \le \tfrac{1}{4}\sqrt{b}$ be a power of two. There is a data structure maintaining a weighted forest with $k$ vertices under the operations \alg{subtree-sum}, \alg{cut}, and \alg{decrement-weight}, provided that all weights are non-negative integers, and the sum $W$ of all weights in the initial forest is at most $2^b-1$.
		
		The data structure requires a certain global table only depending on $k$, which can be computed in time $\fO((3k)^k)$.
		After precomputing that table, initialization takes $\fO(k)$ time, each \alg{subtree-sum} operation takes constant time, each \alg{cut} operation takes $\fO(k)$ time, and all \alg{decrement-weight} operations together take $\fO(W)$ time.
	\end{lemma}
	\begin{proof}
		Let $(F,w)$ denote the current weighted forest, and say $V(F) = [k]$.

		We follow Dietz~\cite{Dietz1989} in using two arrays $A$ and $B$ of length $k$. $A$ stores the actual subtree-sum for each vertex in $F$, but is only updated every $k$ operations (such an update is called a \emph{flush}). $B$ stores for each vertex the number of decrements since the last flush. Since this number can be no more than $k$, the entire array $B$ fits in $k (1 + \log k) \le b$ bits, so it fits into a single machine word.
		
		To adapt the Dietz's idea to trees with \alg{cut}s, we need the following two extra data structures. First, we want to store, for each vertex $v$, the set of vertices $V(F_v)$ in the subtree rooted at $v$ (i.e., the ones whose weights count towards $\alg{subtree-sum}(v)$). We store the mapping $v \rightarrow V(F_v)$ in an array $C$. Each $V(F_v)$ is represented as a bitstring, so $C$ has size $k^2 \le b$ and fits into a single machine word.
		
		Second, to actually obtain the sum $w(F_v)$, we need the \emph{global table} mentioned in the statement of the lemma. We denote it by $Q$ and it maps a pair $(B',U)$ to an integer. Here $B'$ is an array of length $k$, storing $k$ numbers in $[0,k]$ (as $B$ above), and $U$ is a subset of $V(F) = [k]$, stored as a bitstring. Both $B'$ and $U$ fit into single machine words. The value $Q[B',U]$ is $\sum_{i \in U} B'[i]$.
		Note that $Q$ has $(k+1)^k \cdot 2^k$ entries, each entry fits into a machine word, and $Q$ can be computed in time $|Q| \cdot \fO(k) = \fO((3k)^k)$.
		Observe that $Q$ only depends on $k$, as required.
		
		We now show how to implement the operations. We maintain the invariant that $A$ is correct for a previous point in time (the last \emph{flush}), and since then at most $k-1$ \alg{decrement-weight} and no \alg{cut} operations have been performed. In turn, $B[v]$ is the number of $\alg{decrement-weight}(v)$ operations since the last flush.
		
		\begin{itemize}
			\item Initially, we compute $A$ and $C$ in a bottom-up fashion, in $\fO(k)$ time. We initialize $B$ as an all-zero vector, and assume $Q$ has been already precomputed.
			
			\item To answer a $\alg{subtree-sum}(v)$ query, we calculate $A[v] - Q[B, C[v]]$, in constant time.
			
			\item To perform $\alg{decrement-weight}(v)$, we first increment $B[v]$. Then, if $k$ \alg{decrement-weight} operations have been performed since the last flush, we perform a flush as follows. For each $v \in V(F)$, set $A[v] \gets A[v] - Q[B, C[v]]$, and then reset $B$ to the all-zero array. The amortized time for the flushes is clearly constant per \alg{decrement-weight} operation. Since we can perform no more than $W$ \alg{decrement-weight} operations, the total time is at most $\fO(W)$.
			
			\item To perform $\alg{cut}(v)$, we first perform a flush, then recover all vertex weights from $A$, then delete the edge from $v$ to its parent and recompute $A$ and $C$. This requires $\fO(k)$ time, as desired.\qedhere
		\end{itemize}
	\end{proof}
	
	We now show how to iterate \cref{p:subtree-size-small}, again using a cluster decomposition (though with different cluster size than before).

	We introduce a new operation that simplifies the analysis. When maintaining a weighted forest which is initially a tree, the parameterless operation $\alg{root-sum}()$ returns the sum of weights of vertices still connected to the root of the \emph{initial} tree. Note that \alg{root-sum} can and will be implemented with a single call $\alg{subtree-sum}(v)$, where $v$ is the initial tree root (which can be stored at the beginning).

	We present a sequence of data structures $X_1, X_2, \dots$, which support increasingly larger inputs, but also have successively worse running time. For now, we assume the input is a binary tree. For technical reasons related to binarization (\cref{p:binarize}), we consider 0-1-weighted forests instead of unweighted forest. Below, big-$\fO$ notation only hides constants not depending on the values $n$, $b$, or $t$. %
	
	\begin{lemma}\label{p:subtree-size-param}
		Fix a word size $b$ and let $\ell$ be a power of two with $\ell \le \frac{1}{48}\sqrt{b}$. For each parameter $t \in \N_+$, there is a data structure maintaining an 0-1-weighted binary forest, initially a tree, with $n \le \min\{\ell^t,2^b-1\}$ vertices under the operations \alg{cut} and \alg{subtree-sum}.
		
		The data structure requires a certain global table only depending on $\ell$, which can be computed in time $\fO((36\ell)^{12\ell})$.
		After precomputing that table, the initialization takes $\fO(tn)$ time, each \alg{subtree-sum} operation takes $\fO(t)$ time, each \alg{root-sum} operation takes $\fO(1)$ time, and all \alg{cut} operations together take $\fO( t ( n + \ell^2) )$ time.
	\end{lemma}
	\begin{proof}
		The case $t = 1$ follows from \cref{p:subtree-size-small}. We thus focus on the case $t \ge 2$.
		
		Let $(T^0,w_0)$ be the initial 0-1-weighted tree.
		At the start, use \cref{p:cluster-decomp} to compute a cluster decomposition $\fC$ of $(T^0,w_0)$ with $|\fC| \le 6 \ell$, where each cluster has size at most~$\tfrac{n}{\ell}$.
		
		As in \cref{p:tree-sum-red}, we keep the cluster decomposition the same throughout the run of the data structure. At any point, let $(F,w)$ denote the current weighted forest, and let $G$ denote the cluster forest induced by $\fC$. We again maintain $G$ using \cref{p:maintain-cluster-forest}, and also keep the helper data structure of \cref{p:roots} on $F$.
		
		We also introduce a weight function $w'$ on $G$, which is slightly different from the one used in \cref{p:tree-sum-red}. For each cluster $C$, we let $w'(\ub(C))$ be the total weight of all vertices in $C \setminus \{\lb(C)\}$ that are connected to $\ub(C)$. The weight of $\lb(C)$ is always $w(\lb(C))$. Observe that, if $v$ is a boundary vertex of some cluster, then the subtree-sum of $v$ w.r.t.\ $(F,w)$ is precisely the same as the subtree-sum of $v$ w.r.t.\ $(G,w')$.

		The weighted cluster forest $(G,w')$ is maintained using an instance $D$ of the data structure of \cref{p:subtree-size-small} with parameter $k = 12\ell$. To check that all conditions of \cref{p:subtree-size-small} are met, first observe that $G$ has at most $2 |\fC| \le 12 \ell = k \le \frac14 \sqrt{b}$ vertices. Second, the sum of weights $w'(G)$ is precisely the sum of weights $w(F)$, which is at most $|V(F)| = n \le 2^b-1$, as required.
		
		As in the proof of \cref{p:tree-sum-param}, we maintain a \emph{cluster object} for each cluster $C$, which is referenced by a pointer from each of its vertices. The cluster object stores pointers to $\ub(C)$, to $\lb(C)$, and additionally a subtree-size data structure $X_C$ on the subtree $(T^0,w_0)$ induced by $C \setminus \{\lb(C)\}$. This data structure is constructed recursively with parameter $t-1$. This is possible because each cluster has size at most $\tfrac{n}{\ell} \le \ell^{t-1}$.
		
		We now describe how to implement the three operations.
		\begin{itemize}
			\item Consider a call $\alg{subtree-sum}(v)$. If $v$ is a boundary vertex, we return $D.\alg{subtree-sum}(v)$, as explained above. Otherwise, determine the cluster $C$ containing $v$, and check whether $v$ is an ancestor of $\lb(C)$. %
			If $\lb(C) = \bot$ or $v$ is not an ancestor of $\lb(C)$, then we can simply return $X_C.\alg{subtree-sum}(v)$, since all descendants of $v$ are contained in $C$. Otherwise, we return $X_C.\alg{subtree-sum}(v) + D.\alg{subtree-sum}(\lb(C))$. It is easy to see that this is correct.
			
			\item $\alg{root-sum}()$ is implemented simply by calling $D.\alg{subtree-sum}(u)$ on the initial tree root $u$ (which can be stored at the start). Observe that $u$ is a boundary vertex, and thus indeed contained in $D$.
			
			\item Finally, consider a call $\alg{cut}(v)$. First update $G$ using \cref{p:maintain-cluster-forest}; if this results in an edge removed from $G$, mirror this with $D.\alg{cut}$. Let $C$ be the cluster containing~$v$. If $v = \ub(C)$ or $v = \lb(C)$, then no change to $X_C$ is needed. Otherwise, we call $X_C.\alg{cut}(v)$. This may disconnect some vertices of $C \setminus \{\lb(C)\}$ previously connected to $u = \ub(C)$. Thus, we may need to change $w'(u)$. Since $u$ is the root of $T^0[C]$, we can determine the new value with $X_C.\alg{root-sum}()$. Observe that $w'(u)$ can only decrease; thus, we can update $w'(u)$ using an appropriate number of calls to $D.\alg{decrement-weight}(u)$.
		\end{itemize}
		
		We now analyze the running time. The global table needs only be computed once (for $D$ and all the recursive data structures, since $\ell$ is fixed).
		
		The initialization takes $\fO(n)$ time for everything besides the data structures $X_C$. By induction, the total time is $\fO(tn)$.
		
		For \alg{subtree-sum}, observe that we perform constant work plus (possibly) one recursive call to $X_C.\alg{subtree-sum}$. Again by induction, this takes $\fO(t)$ time.
		Then \alg{root-sum} works like \alg{subtree-sum}, but the $X_C.\alg{subtree-sum}$ call is always omitted, so the time is constant.
		
		For \alg{cut}, we perform constant work, plus at most one $X_C.\alg{root-sum}$ call, one $X_C.\alg{cut}$ call, one $D.\alg{cut}$ call, and some number of $D.\alg{decrement-weight}$ calls. The $X_C.\alg{root-sum}$ takes constant time by induction. By \cref{p:subtree-size-small}, each $D.\alg{cut}$ call takes $\fO(k) = \fO(\ell)$ time, and there are at most $|\fC|-1 \le k - 1 \le \fO(\ell)$ such calls, so the total running time is $\fO(\ell^2)$. The total running time of $D.\alg{decrement-weight}$ is $\fO(W)$, where $W$ is the total weight of vertices in $G$, which is at most $n$ by definition.
		
		This means the total time for all \alg{cut}s is $\fO(n + \ell^2)$ plus the time for all $X_C.\alg{cut}$ calls. Since the total number of all vertices in all clusters is precisely $n$, we obtain the bound $\fO( t(n + \ell^2) )$ by induction.
	\end{proof}
	
	We are now ready to prove:
	
	\restateSubtreeSize*
	\begin{proof}
		We assume the machine word size is at least $b = 1 + \lceil \log(n+1) \rceil$.\footnote{The assumption that $b \ge \Omega(\log n)$ is standard, and if it is a constant factor less than necessary, we use the usual simulation of $b$-bit-word operations with a constant number of words.}
		Let $b' = \lceil \log n \rceil \le b$.
		As usual, we can assume the initial weighted forest is a tree $(T^0,w_0)$. First, binarize the tree with \cref{p:binarize} to obtain $(T^1, w_1)$. Now use \cref{p:subtree-size-param} with parameters $\ell = \lfloor \tfrac{1}{48} \sqrt{b'} \rfloor$ and $t = \lceil \log_\ell (2n) \rceil$. Note that $|V(T^1)| \le 2n < 2^b-1$ and $|V(T^1)| \le \ell^t$, so the conditions for \cref{p:subtree-size-param} are satisfied. The total running time is $\fO( t (m+n + \ell^2) )$. Observe that $\ell^2 \le b' \le n$. Moreover, we have $t \le \fO(\tfrac{\log n}{\log\log n})$ by an easy calculation. The statement follows.
	\end{proof}
	
	\bibliographystyle{alphaurl}
	\bibliography{tree-sum}
	
	\appendix
	
	\crefalias{section}{appendix}

	\section{Optimal data structures for small trees}\label{sec:opt-small}

	In this section, we give the details on our procedure to efficiently compute (asymptotically) optimal tree-sum data structures for very small forests.
	Recall that $\OPT(F,m)$ is the minimum number of group operations required by any data structure to process $m$ operations with the initial forest $F$, in the worst case taken over all initial weight assignments and operation sequences.
	
	As in \cref{sec:tree-sum-optimal}, we say that a \emph{tree-sum data structure} supports \alg{cut}, \alg{update-weight}, \alg{tree-sum}, and allows for auxiliary vertices.
	It is crucial here that any single data structure is required to work with any given group, it does not know the group, and cannot make decisions based on memory cells that store group elements. For example, comparisons with zero are not allowed; c.f.\ Preliminaries (\cref{sec:prelims}). We will prove:

	\restateOptSmall*
	
	\subsection{Computation trees}
	
	It will be useful to define the following alternative model of computation, similar to decision trees. Fix an initial forest $F$ with $n$ vertices, potentially containing some auxiliary vertices. A \emph{computation tree} for $F$ is a (rooted) tree $C$. Each edge of $C$ is labeled with an \emph{operation}, which is $\alg{cut}(v)$, $\alg{tree-sum}(v)$, or $\alg{update-weight}(v)$ for some vertex $v \in V(F)$. Note the missing weight parameter for \alg{update-weight}.
	
	Each node $x$ is labeled with a sequence $I_x$ of \emph{instructions}. Informally, each instruction is an addition or subtraction of weights or memory registers, the result of which is stored in a specified memory register.
	More precisely, each instruction has one of the following forms. Here $X$ and $Y$ each represent either $R[j]$ for some $j \in \N_+$ or $W[v]$ for some $v \in V(F)$.
	\begin{itemize}
		\item $R[i] \gets X + Y$ for $i \in \N_+$
		\item $R[i] \gets X - Y$ for $i \in \N_+$
	\end{itemize}
	(Sequences of this form are sometimes called \emph{straight-line programs} in the literature.)
	
	If $x$ is not the root and the edge to its parent is labeled $\alg{tree-sum}(v)$ for some $v$, then $x$ is additionally labeled with a \emph{return value} $E_x$, which is again either $R[j]$ for some $j \in \N_+$ or $W[v]$ for some $v \in V(F)$.
	
	Informally, the semantics of this model are as follows.
	The $R[i]$ are freely usable registers which store group elements, and each $W[v]$ represents the current weight of a vertex $v$.
	The computation starts at the root $r$ of $C$ and initially executes the sequence $I_r$ of instructions (potentially storing some values in registers $R[i]$).
	When an operation is performed, follow the respective edge to a child node $y$, and execute $I_y$. If the operation is $\alg{update-weight}(v,t)$, first update $W[v] \gets t$ and then follow the edge $\alg{update-weight}(v)$ and execute $I_y$. If the operation is \alg{tree-sum}, return $E_y$ at the end.
	
	Note that a computation tree is specific to a forest, and moreover can only handle a finite number of operations.
		
	We now formally define the correctness of a computation tree. First, we say that $C$ is \emph{valid} if:
	\begin{itemize}
		\item If an~edge $e$ is labeled $\alg{cut}(v)$, then $v$ is not auxiliary nor a~root of $F$, and no ancestor edge of $e$ is labeled $\alg{cut}(v)$.
		\item For each node $x \in V(C)$ of depth at most $\height(C)-1$, and each $v \in V(F)$, there are edges labeled $\alg{cut}(v)$, $\alg{tree-sum}(v)$, and $\alg{update-weight}(v)$ from $x$ to a child, if allowed by the previous condition.
		\item No two edges from the same parent share a label.
	\end{itemize}
	Note that this implies that all leaves of $C$ have the same depth.
	
	Fix a valid computation tree $C$ of height $m$ on $F$, a group $G$, an initial weight function $w \colon V(F) \rightarrow G$, and a sequence $\sigma$ of at most $m$ operations.
	The \emph{output} of $C$ on $(G,w,\sigma)$ is a sequence of values $g_1, g_2, \dots, g_k \in G$, one for each \alg{tree-sum} operation in $\sigma$, computed as follows.
	\begin{itemize}
		\item At the start, initialize arrays $R$ and $W$ with $R[i] \gets 0$ for each $i \in \N_+$, and $W[v] \gets w(v)$.
		Also initialize a pointer $x$ to the root of $C$.
		Execute the instructions $I_x$ in the obvious way.
		\item For each operation, first follow the corresponding edge from $x$ to a child~$y$ (observe that such an edge exists by definition). If the operation is $\alg{update-weight}(v,t)$, set $W[v] \gets t$. Then, execute $I_y$. if the operation is \alg{tree-sum}, add the value $E_y$ to the output. Finally, set $x \gets y$.
	\end{itemize}
	
	We call $C$ \emph{correct for $(G,w,\sigma)$} if applying the operations $\sigma$ to $(F,w)$ in a tree-sum data structure yields the same output. Further, we call $C$ \emph{correct} if $C$ is correct for each $(G,w,\sigma)$ where $|\sigma| \le m = \height(C)$.
	
	We now define an analog of the time complexity of a computation tree $C$.
	Let the \emph{instruction depth} $i_x$ of $x$ be $\sum_{y \in Y} |I_y|$, where $Y$ is the set of ancestors of $x$, including $x$ itself. The \emph{maximum instruction depth} of $C$ is $\MID(C) = \max_{x \in V(C)} i_x$.
	
	For a forest $F$ and $m \in \N_+$, define $\OPTCT(F,m)$ as the minimum $\MID(C)$ among all correct computation trees for $F$ of height $m$.
	
	We call $C$ \emph{succinct} if all indices $i,j$ of the registers used in instructions of $I_x$ are at most the instruction depth $i_x$ of $x$. Observe that every computation tree can be converted into a succinct computation tree with the same time complexity, by reassigning register indices based on their first use.
	
	With the next two lemmas, we show that (succinct) computation trees are equivalent to tree-sum data structures in terms of time complexity, at least when we restrict ourselves to a specific initial forest $F$. In particular, we will show that $\OPTCT(F,m) = \OPT(F,m)$.
	Below, we say a tree-sum data structure is \emph{for} a forest $F$ if it only accepts $F$ as the initial forest (with arbitrary weights).
	
	\begin{lemma}\label{p:bf:tree-to-ds}
		Let $C$ be a correct succinct computation tree of depth $m$ for a forest $F$ of size $n$ with time complexity $t$. Then, there exists a tree-sum data structure for $F$ that performs $m$ operations in time $\fO(n+m+\MID(C))$, using $\MID(C)$ additions and subtractions.
	\end{lemma}
	\begin{proof}
		When defining the \emph{output} of a computation tree above, we essentially already constructed a tree-sum data structure based on $C$.
		The number of additions and subtractions is clearly $\MID(C)$.
		The time overhead is linear for initialization and constant for each operation or instruction.
	\end{proof}
	
	\begin{lemma}\label{p:bf:ds-to-tree}
		Let $F$ be a forest of size $n$ and let $D$ be a tree-sum data structure for $F$. Then, there exists a computation tree $C$ for $F$ with $\MID(C) \le \Time(D,F,m)$.
	\end{lemma}
	\begin{proof}
		The crucial observation is that, for a fixed forest $F$ and a fixed sequence of operations, the data structure $D$ performs precisely the same Word RAM operations; only the contents of the memory cells holding group elements depend on the initial weights and weight updates.
		
		We now proceed with the formal proof. First, construct a valid computation tree $C$ for $F$ of height $m$, ignoring the instruction sequences $I_x$ and return values $E_x$ for now.
		Observe that there is a unique such computation tree.
		
		Construct each $I_x$ and $E_x$ as follows. Start $D$ with some dummy weights, and record all additions and subtractions performed during preprocessing. Put all these additions and subtractions into $I_r$ for the root $r$ of $C$. Then, record the state of the algorithm (excluding contents of special registers with group elements), and (temporarily) store it as $s_r$.
		
		Now, for every possible edge labeled $\ell$ from $x$ to a child $y$, we do the following: Make a copy of $D$ with state $s_x$.
		On this copy, perform the operation $\ell$; if it is \alg{update-weight}, use a dummy weight as the second parameter. Again, record all additions and subtractions. Build the corresponding instructions, put them into $I_y$, and store the final state as $s_y$.
		If $\ell = \alg{tree-sum}(v)$, also record the memory cell returned by $D$ and store it as $E_y$.
		
		It is easy to see that $C$ is correct if $D$ is correct.
		Moreover, the maximum instruction depth of $C$ is precisely the maximum number of additions and subtractions that $D$ performs for a sequence of $m$ operations. That is, $\MID(C) \le \Time(D,F,m)$, as desired.
	\end{proof}
	
	\begin{corollary}
		For every forest $F$ and every $m \in \N_+$, we have $\OPTCT(F,m) = \OPT(F,m)$.
	\end{corollary}
	
	\subsection{Computing optimal computation trees}\label{sec:opt-small-compute}
	
	We now show how to find an optimal computation tree for a given forest $F$ and number $m$ of operations by brute force.
	First, we bound the number of computation trees that we need to consider.
	
	\begin{lemma}\label{p:bf:tree-enum}
		Let $F$ be a forest of size $n \ge 2$ and let $m, d \in \N_+$.
		The number of valid and succinct (but not necessarily correct) computation trees for $F$ with height $m$ and maximum instruction depth $d$ is at most $2^{n^{\fO(m)} \cdot d}$,
		and they can be enumerated within the same time bound.
	\end{lemma}
	\begin{proof}
		Let $C$ be a succinct computation tree for $F$. There are $3n$ possible edge labels (three operation types applied to $n$ different vertices), so each node in $C$ has at most $3n$ children. Since the height of $C$ is precisely $m$, we have $|V(C)| \le n^{\fO(m)}$.
		
		Since the maximum instruction depth is at most $d$ by assumption, in particular each instruction sequence $I_x$ has length at most $d$. Thus, the total number of instructions in $C$ is $|V(C)| \cdot d$. Since $F$ is succinct, it only uses register indices up to $d$. The weight array $W$ is accessed with vertex indices at most $n$. Thus, each instruction can be encoded using at most $2 \log \max\{n,d\}$ bits, so the number of bits to encode $C$ is
		\[ n^{\fO(m)} \cdot d \cdot \log \max\{n,d\} = n^{\fO(m)} \cdot d. \]
		To enumerate $C$, simply enumerate all bitstrings of this length, and check for each whether it corresponds to a valid and succinct computation tree. The time to do the latter is clearly polynomial in the bitstring length, and thus negligible.
	\end{proof}
	
	Second, we show how to check the correctness of a computation tree and compute its maximum instruction depth (in order to select the best one later on).
	
	\begin{lemma}\label{p:bf:tree-check}
		Let $F$ be a forest of size $n \ge 2$, and let $m \in \N_+$.
		Let $C$ be a valid and succinct computation tree for $F$ of height $m$.
		Then, we can compute $\MID(C)$ and check whether $C$ is correct in time $n^{\fO(m)} \cdot \MID(C)$.
	\end{lemma}
	\begin{proof}
		Computing $\MID(C)$ is trivially done in time proportional to the total number of vertices and instructions in $C$, which is $|V(C)| \cdot \MID(C) \le n^{\fO(m)} \cdot \MID(C)$.
		
		Let $n'$ be the number of non-auxiliary vertices in $F$.
		To perform correctness check for operation sequences of length at most $m$, we use the commutative group $G = \mathbb{Z}^{n'+m}$; let $e_i = (0, 0, \ldots, 0, 1, 0, \ldots, 0) \in G$ for $i \in [n'+m]$ denote the canonical base vector that contains a~single value $1$ at position $i$.
		Assume that the set of non-auxiliary vertices in $F$ is indexed by $[n']$ and set the initial weight of vertex $i \in [n']$ to $e_i$.
		The $i$th call $\alg{update-weight}(\cdot, w_i)$ updates the weight of a~specified vertex to $w_i = e_{n' + i}$.
		Whenever the computation of $C$ returns a sum corresponding to some subgraph $T$ of the initial forest $F$, we check that it is a vector with only zeros and ones, where all ones correspond precisely to the current weights of the vertices in $V(T)$.
		Thus, the simulation of $C$ on a~given operation sequence can be done in time polynomial in $n$ and $m$.
		
		If any of tested sequences fails the simulation (i.e., it returns an~incorrect result for any query), then $C$ is clearly incorrect.
		Otherwise, it is easy to see that the data structure is correct for any commutative group $G$.

		Observe that we test $C$ on at most $n^{\fO(m)}$ operation sequences. Since a single addition or subtraction in $C$ can be computed in $\fO(m+n)$ time, the time to test a single operation sequence is at most $(m+n) \cdot \MID(C)$.
		The overall time is thus $n^{\fO(m)} \cdot (m+n) \cdot \MID(C) = n^{\fO(m)} \cdot \MID(C)$, as desired.
	\end{proof}
	
	\begin{corollary}\label{p:bf-opt-single-m}
		Let $F$ be a forest of size $n \ge 2$, and let $m \in \N_+$. Then, we can compute a correct computation tree $C$ for $F$ with height $m$ and $\MID(C) = \OPT(F,m)$ in time $2^{n^{\fO(m)}}$.
	\end{corollary}
	\begin{proof}
		From \cref{p:tree-sum-logn,p:bf:ds-to-tree}, it follows that $\OPT(F,m) \le \fO(m + n \log n)$.
		Enumerating all possible valid and succinct computation trees with maximum instruction depth $d \le \fO(m + n \log n)$ can be done within the time bound above using \cref{p:bf:tree-enum}.
		We can filter out incorrect computation trees and select one $C$ with $\MID(C) = \OPTCT(F,m) = \OPT(F,m)$ using \cref{p:bf:tree-check}. The overall running time is as stated.
	\end{proof}
	
	By now, we have shown how to compute optimal computation trees when the number $m$ of operations is known. We now proceed to build an (asymptotically) optimal data structure for arbitrary, initially unknown $m$, thereby proving \cref{p:opt-small}.
	
	The idea is to start with an optimal data structure $C_m$ for some constant $m$. As soon as we receive the $(m+1)$th instruction, we switch to an optimal data structure $C_{m'}$ with a larger $m'$, and so on. It is tempting to just double $m$ in every step, and try to apply a geometric sum argument.
	However, we do not actually know enough about the function $\OPT(F,m)$ to make this work.
	Instead, we start with a constant \emph{time budget} $t$, compute the largest $m$ such that $\OPT(F,m) \le t$, and use $C_m$. Whenever we go over the limit of $m$ operations, we double the time budget, and recompute $m$. With this, the geometric sum argument works.
	
	We remark that Pettie and Ramachandran use first strategy of doubling $m$ in their optimal data structure for the \emph{split-findmin} problem~\cite[Appendix B]{PettieRamachandran2005}. It appears that the same problem (and solution) applies to that problem, so their proof is erroneous, but fixable.
	
	\restateOptSmall*
	\begin{proof}
		For each $1 \le m \le n^2$, compute a correct computation tree $C_{m}$ of height $m$ and with $\MID(C_m) = \OPT(F,m)$, using \cref{p:bf-opt-single-m}.
		Note that we can at the same time compute and store the value $t(m) \coloneq \MID(C_m) + n + m$ for each $m$.
		Observe that $t_1 \le 2n+1 \le 3n$, since $t_1 \le \OPT(F,1) \le n$ (a single query can be answered by adding all required vertex weights by brute force).
		Turn each $C_{m}$ into a data structure $D_{m}$ using \cref{p:bf:tree-to-ds}; the data structure has total running time $\fO(\OPT(F,m) + n+m) = \fO(t(m))$ for any sequence of at most $m$ operations.
		
		Suppose for now that the (initially unknown) number of operations is at most $n^2$.
		Our data structure now works as follows:
		Set $t^* = 3n$, and set $m^*$ to the maximum $m \leq n^2$ such that $t(m) \le t^*$. Observe that $m^*$ is well defined, since $t_1 \le 3n$, as shown above.
		Also store the list $\sigma$ of operations received so far, and initialize the data structure $D \gets D_{m^*}$.
		
		As long as $|\sigma| \le m^*$, we delegate all operations to $D$.
		Now suppose we receive the $(m^*+1)$th instruction. Then, we set $t^* \gets 3t^*$, recompute $m^*$ to be again the maximum $m$ such that $t(m) \le t^*$, reinitialize a new data structure $D \gets D_{m^*}$, and then replay all operations so far on $D$. We continue delegating to $D$, while occasionally rebuilding.
		
		Let now $m$ be the actual number of operations, and let $m^*_i$ be the number of operations before the $i$th update was triggered.
		The number of rebuildings we have performed is $q = \log_3 \tfrac{t^*}{3n}$; let $m^*_{q+1}$ denote the final value of $m^*$.
		Observe that $m^*_1 \leq \ldots \leq m^*_{q+1}$.
		Moreover, $m^*_q < m \leq m^*_{q+1}$ by construction.
		If $q = 0$, then $t(m) \leq t(m^*_1) \leq t^* = 3n$ and so the total running time of $D$ is $\fO(n)$.

		Now assume $q \geq 1$. Then $m > m^*_q$, so in particular $t(m) > \tfrac13 t^*$ since $m^*_q$ is maximal such that $t(m^*_q) \leq \tfrac13 t^*$.
		Thus $t^* < 3t(m) \leq 3(\OPT(F, m) + n + m)$.
		Also $t(m^*_{i}) \leq 3^{i-q-1} t^*$ for all $i \in [q + 1]$ by construction.

		The overall running time of our data structure is clearly $\sum_{i=1}^{q+1} \fO( t(m^*_i))$.
		The required bound on the running time follows from a simple calculation:

		\[
			\sum_{i=1}^{q+1} t(m^*_i) \leq \sum_{i=1}^{q+1} 3^{i-q-1} t^* \leq \tfrac43 t^* < 4(\OPT(F, m) + n + m).
		\]

		We now discuss how to handle the case where the final number of operations $m$ is greater than $n^2$.
		For the first $T$ operations, we proceed as above; this takes time $\fO(\OPT(F, n^2) + n^2)$.
		Observe that by \cref{p:tree-sum-logn} we have in particular $\OPT(F, n^2) \leq \fO(n \log n + n^2) = \fO(n^2)$.
		Therefore, the first $n^2$ operations are processed in total time $\fO(n^2)$.
		After $n^2$ operations, we reinitialize $D$ to the data structure of \cref{p:tree-sum-logn}, replaying the list $\sigma$ of received operations one final time; and afterwards, we relay all subsequent operations to that data structure.
		The data structure takes time $\fO(n \log n + m) = \fO(m)$ to process all $m$ queries.
		Adding these two terms yields the running time of $\fO(m) \le \fO(\OPT(F,m) + n + m)$.
	\end{proof}

	\section{Proof of \cref{t:adapted-prefix-sum-lb}}
	\label{sec:adapted-prefix-sum-lb}

	For completeness, we include the proof of \cref{t:adapted-prefix-sum-lb}.
	We recall the statement below:
	\restateAdaptedPatrascuDemaine*
	
	While the proof can be deduced from the techniques of the original work of P{\u{a}}tra{\c{s}}cu and Demaine~\cite{PatrascuDemaine2006}, we find it more convenient to use the framework laid out by P{\u{a}}tra{\c{s}}cu's thesis~\cite[Chapter 3]{Patrascu2008}.
	His lower bound assumes that a~data structure maintains an array $A$ of $n$ integers from $0$ to $U - 1$, where $U \geq n^{\Omega(1)}$, and the sequence of operations to $A$ is governed by a~permutation $\pi$ on $[n]$.
	The sequence of operations is as follows: For each $i \in [n]$, the \alg{partial-sum} of initial $\pi(i)$ elements of $A$ is queried first, and then $A[\pi(i)]$ is \alg{update}d to some new value from $0$ to $U - 1$, chosen uniformly at random.
	Let us call such a~sequence of operations \emph{single-epoch-based}, and each data structure executing such a~sequence a~\emph{single-epoch-based data structure}.
	
	In order to give a~lower bound on the number of cell probes required to execute this sequence of operations, we introduce a~\emph{lower-bound tree} $\mathcal{T}$: a~binary tree with $n$ leaves.
	Each non-leaf vertex has a~left and a~right child.
	The leaves of $\mathcal{T}$ are identified with the indices $1, 2, \ldots, n$ of the array $A$, from left to right.
	For each internal vertex $v$ of $\mathcal{T}$, we define the \emph{interleaving} at $v$, denoted $\mathrm{IL}(v)$, as follows.
	Let $L_v$ be the set of values $\pi(i)$ for all leaves $i$ in the left subtree of $v$, and $R_v$ be the analogous set for the right subtree of $v$.
	The interleaving at $v$ is defined as the number of pairs $(x,y)$ with $x \in L_v$, $y \in R_v$, where $x$ comes immediately before $y$ in the sorted sequence of $L_v \cup R_v$.
	The \emph{total interleaving} of $\pi$ (with respect to $\mathcal{T}$) is defined as the sum of interleavings at all internal vertices of $\mathcal{T}$.
	We then have:
	
	\begin{theorem}[{\cite[Theorem 3.6]{Patrascu2008}}]
		\label{t:orig-patrascu-lb}
		Consider any cell-probe single-epoch-based data structure for \alg{partial-sum} on a~length-$n$ array with entries from $0$ to $U - 1$ ($U \geq n^{\Omega(1)}$) that may use Las Vegas randomization.
		The structure executes the sequence of operations governed by a~permutation $\pi$ on $[n]$.
		The expected number of cell probes performed by the data structure is at least $\Theta\left(\sum_{v \in \mathcal{T}} \mathrm{IL}(v) - n\right)$.
	\end{theorem}
	
	Note that \cref{t:orig-patrascu-lb} holds for any permutation $\pi$ and lower-bound tree~$\mathcal{T}$.
	P{\u{a}}tra{\c{s}}cu then concludes his lower bound proof by showing that a~random permutation $\pi$ has a~large expected total interleaving:
	
	\begin{lemma}[{\cite[Claim 3.7]{Patrascu2008}}]
		\label{l:random-perm-interleaving}
		Let $\pi$ be a~permutation on $[n]$ chosen uniformly at random.
		Then, the expected total interleaving of $\pi$ (with respect to some fixed lower-bound tree $\mathcal{T}$ with $n$ leaves) is $\Omega(n \log n)$.
	\end{lemma}
	
	We now turn to the proof of \cref{t:adapted-prefix-sum-lb}.
	Consider a~cell-probe data structure $D$ processing epoch-based sequences of operations as defined in the theorem statement.
	We construct a~single-epoch-based data structure $D'$ as follows.
	Suppose $D'$ is initialized with an array $A'$ of size $n' = En$ with initial values from $0$ to $U$.
	For $i \in [n]$, let the $i$th \emph{block} of $A'$ be the subarray $A'[(i-1)E + 1 \ldots iE]$.
	Intuitively, the $i$th block of $A'$ will be compressed to the element $A[i]$ of the array $A$ of size $n$ used by $D$; and during each of $E$ epochs of $D$, one element of each block will be updated.
	
	We now choose a~permutation $\pi'$ on $[n']$ governing the sequence of operations to $D'$.
	Let $\pi$ be a~permutation on $[n]$ chosen uniformly at random.
	Then:
	\[
	\pi'((p - 1)n + q) = (\pi(q) - 1)E + p \quad \text{for all } p \in [E], q \in [n];
	\]
	that is, in the $q$th operation of the $p$th epoch, we \alg{update} the $p$th element of the $\pi(q)$th block of $A'$.
	Hence the first epoch \alg{update}s the first element of each block, the second epoch \alg{update}s the second element of each block, and so on.
	The order of blocks being updated in each epoch is identical and determined by the random permutation $\pi$.
	We observe that $\pi'$ has large expected interleaving:
	
	\begin{lemma}
		\label{l:adapted-random-perm-interleaving}
		The expected total interleaving of $\pi'$ (with respect to some fixed lower-bound tree $\mathcal{T}'$ with $n'$ leaves) is $\Omega(n' \log n')$.
	\end{lemma}
	\begin{proof}
		Let $\mathcal{T}$ be the lower-bound tree with $n$ leaves from \cref{l:random-perm-interleaving} used to analyze~$\pi$.
		We construct $\mathcal{T}'$ by creating $E$ copies $\mathcal{T}_1, \ldots, \mathcal{T}_E$ of $\mathcal{T}$, and relabeling the leaf $q$ of $\mathcal{T}_p$ with $(p - 1)n + q$; then, we link all the copies of $\mathcal{T}'$ to form a~new binary tree $\mathcal{T}$ containing the subtrees $\mathcal{T}_1, \ldots, \mathcal{T}_E$ in the left-to-right order.
		
		By the definition of $\pi'$, the value $\sum_{v \in \mathcal{T}_p} \mathrm{IL}(v)$ is equal to the total interleaving of $\pi$ with respect to $\mathcal{T}$ for each $p \in [E]$; this is because within each epoch, the order of blocks being updated is identical and determined by $\pi$.
		Therefore, the total interleaving of $\pi'$ with respect to $\mathcal{T}'$ is equal to $E$ times the total interleaving of $\pi$ with respect to $\mathcal{T}$.
		The lemma now follows from \cref{l:random-perm-interleaving} and the fact that $\Omega(En \log n) = \Omega(n' \log n')$; note that here we use that $E \leq n^{\fO(1)}$.
	\end{proof}
	
	\cref{t:orig-patrascu-lb,l:adapted-random-perm-interleaving} now imply that the expected number of cell probes performed by $D'$ is at least $\Omega(n' \log n')$.
	It remains to implement $D'$ using $D$.
	To this end, we first notice that the \alg{partial-sum} queries \emph{within each block} can be implemented efficiently:
	\begin{lemma}
		\label{l:prefix-sum-blocks}
		Suppose $B$ is a~dynamic integer array of length $s$ undergoing $s$ \alg{update}s, where the $p$th \alg{update} modifies $B[p]$.
		A~data structure reporting partial sums of $B$ can be initialized in $\fO(s)$ time and perform each query in $\fO(1)$ time.
	\end{lemma}
	\begin{proof}
		On initialization, we copy $B$ to $B_0$ and compute all partial sums of $B_0$ in $\fO(s)$ time; this allows us to determine the sum of any interval $B_0[i \ldots j]$ in $\fO(1)$ time.
		We also incrementally maintain \emph{updated} partial sums: Let $B_1$ be an initially zeroed array of length $s$.
		After $p$ \alg{update}s, each $B_1[i]$ for $i \in [p]$ contains the current sum of $i$ first elements of $B$.
		When $B[p]$ is \alg{update}d, $B_1[p]$ is computed from $B_1[p-1]$ (if it exists) and the new value of $B[p]$ in $\fO(1)$ time.
		Finally, we keep a counter $p$ of the number of \alg{update}s performed so far. 
		A~\alg{partial-sum} query on $B$ at index $i$ can be computed in $\fO(1)$ time as follows: If $i \le p$, return $B_1[i]$. Otherwise, return $B_1[p] + B_0[p+1 \dots i]$.
	\end{proof}
	
	Now, on initialization of $D'$, we instantiate a~data structure $X_1, \ldots, X_n$ of \cref{l:prefix-sum-blocks} for each of the $n$ blocks of $A'$.
	We would like to initialize an array $A$ of size $n$ for $D$, where $A[i]$ contains the sum of elements in the $i$th block of $A'$.
	However, this poses an~issue: the elements of $A$ may be as large as $(U - 1) \cdot E$, while $D$ only supports elements up to $U - 1$.
	To resolve this, we represent each element of $A$ in base $U$. %
	More precisely, let $c \in \N$ be a~constant such that $U^c > (U - 1) E$.
	We actually initialize $c$ arrays $A_0, A_1, \ldots, A_{c-1}$ of size $n$ for $D$, preserving the following invariant: for each $i \in [n]$, the sum of elements in the $i$th block of $A'$ is precisely $\sum_{j=0}^{c-1} A_j[i] \cdot U^j$.
	For convenience, in the following description, we let $A$ denote the collection of arrays $A_0, A_1, \ldots, A_{c-1}$ (where each array is handled by a~separate instance of $D$), and whenever we speak of updating an element of $A$ or querying a~partial sum of $A$, we actually mean updating or querying all arrays $A_0, A_1, \ldots, A_{c-1}$ appropriately. %
	
	To process the $((p - 1)n + q)$th \alg{update} of $D'$, which changes the $p$th element of the $\pi(q)$th block of $A'$, we first perform the \alg{update} on $X_{\pi(q)}$, and then replace $A[\pi(q)]$ with the altered sum of elements in the $\pi(q)$th block of $A'$.
	This involves one \alg{update} to $X_{\pi(q)}$, $c$ \alg{update}s to the arrays in $A$, and additional constant overhead.
	To process the $((p - 1)n + q)$th \alg{partial-sum} query of $D'$, which requests the sum of the first $(\pi(q) - 1)E + p$ elements of $A'$, we first query $A$ for the sum of the first $\pi(q) - 1$ elements (i.e., the total sum of initial $\pi(q) - 1$ blocks of $A'$), and then query $X_{\pi(q)}$ for the sum of the first $p$ elements of the $\pi(q)$th block, adding these to the result.
	Similarly, this involves one query in $X_{\pi(q)}$ and $c$ queries in $A$, plus constant overhead.
	
	Summing up, all $\Theta(n')$ operations of $D'$ can be reduced to:
	\begin{itemize}
		\item $\Theta(cn') = \Theta(n')$ operations on $D$,
		\item $\Theta(n')$ operations on the data structures $X_1, \ldots, X_n$ of \cref{l:prefix-sum-blocks}, each taking $\fO(1)$ time, and
		\item additional $\Theta(n')$ overhead.
	\end{itemize}
	
	Hence, by the lower bound on $D'$, the total number of cell probes performed by $D$ is at least $\Omega(n' \log n' - n') = \Omega(n' \log n')$.
	This concludes the proof.

\end{document}